\documentclass[11pt]{article}
\usepackage{amsmath,amsthm,amssymb}
\usepackage{mathtools}
\usepackage{cite}  
\DeclarePairedDelimiter{\ceil}{\lceil}{\rceil}
\DeclareMathOperator*{\argmax}{arg\,max}
\usepackage{graphicx}
\usepackage{hyperref}
\usepackage[margin=1in]{geometry}
\usepackage{booktabs}
\usepackage{url}

\newtheorem{theorem}{Theorem}[section]
\newtheorem{proposition}[theorem]{Proposition}

\newtheorem{lemma}[theorem]{Lemma}
\theoremstyle{definition}

\newtheorem{example}[theorem]{Example}
\theoremstyle{remark}
\newtheorem{remark}[theorem]{Remark}

\title{The Probabilistic Foundations of Surveillance Failure:\\ From False Alerts to Structural Bias}

\author{Marco Pollanen\\
\small Department of Mathematics, Trent University\\
\small Peterborough, ON K9L 0G2, Canada\\
\small \texttt{marcopollanen@trentu.ca}}

\date{}

\begin{document}

\maketitle

\begin{abstract}
Forensic statisticians have long debated whether searching large DNA databases undermines evidential value of matches. Modern surveillance faces an exponentially harder problem: screening populations across thousands of attributes using threshold rules. Intuition suggests requiring many coincidental matches should make false alerts astronomically unlikely. This intuition fails. Consider a system monitoring 1,000 attributes, each with 0.5 percent innocent match rate. Matching 15 pre-specified attributes has probability $10^{-35}$, one in 30 decillion, effectively impossible. But operational systems may flag anyone matching \emph{any} 15 of the 1,000. In a city of one million innocents, this produces about 226 false alerts. A seemingly impossible event becomes guaranteed. This is a mathematical consequence of high-dimensional screening, not implementation failure. We identify fundamental probabilistic limits on screening reliability. Systems undergo sharp transitions from reliable to unreliable with small data scale increases, a fragility worsened by data growth and correlations.  As data accumulate and correlation collapses effective dimensionality, 
systems enter regimes where alerts lose evidential value even when individual coincidences remain vanishingly rare. This framework reframes the DNA database controversy as a regime shift. Unequal surveillance exposures magnify failure, making ``structural bias'' mathematically inevitable. Beyond a critical scale, failure cannot be prevented through threshold adjustment or algorithmic refinement.
\end{abstract}



\section{Introduction}

\subsection{From Database Search to Threshold Screening}

For over two decades, forensic statisticians have debated a fundamental question: does searching a DNA database of a million profiles weaken the evidential value of a match? Stockmarr \cite{stockmarr1999} argued that database searching increases coincidental match probability, while Balding \cite{balding2002} maintained that likelihood ratios preserve evidential weight regardless of search method. This controversy (concerning a \emph{single} database with \emph{exact} profile matching) remains unresolved \cite{storvik2006,kaye2013}.

Modern surveillance systems face an exponentially harder problem. Rather than searching one database for exact matches, they monitor populations across \emph{thousands} of attributes (locations visited, transactions made, communications sent) and flag individuals whose patterns match \emph{sufficiently many} criteria using a threshold rule rather than exact matching. The intuition that ``requiring many coincidences makes false matches astronomically unlikely'' fails catastrophically in this regime.

Consider a concrete example that reveals the nature of this failure. A surveillance system monitors $k=1000$ binary indicators, each with innocent match probability $p=0.005$ (half a percent chance), and flags individuals matching at least $m=15$ attributes. At first glance, this threshold seems extraordinarily conservative. Requiring an exact match across 15 \emph{specific} predetermined attributes would have probability $(0.005)^{15} \approx 3 \times 10^{-35}$, one in 30 decillion. Such exact-match odds suggest false alerts should be astronomically impossible, even across populations of billions.

But operational systems do not require exact matches across 15 predetermined attributes. Instead, they may flag individuals matching \emph{at least 15 of any 1000 attributes}. This seemingly subtle distinction changes everything.
With $\lambda = kp = 1000 \times 0.005 = 5$ expected chance matches per 
innocent person and threshold $m=15$ (three times the expected value), 
direct calculation gives a per-person false alert probability of
\begin{equation*}
q := \Pr(\text{Poisson}(5) \ge 15) \approx 2.26 \times 10^{-4} = 0.0226\%.
\end{equation*}
In a city of $n = 10^6$ innocents, this produces 
approximately $nq \approx 226$ false alerts, and the probability of at least one
false alert is
\[
\Pr(\text{at least one false alert})
  \approx 1 - e^{-226}
  \approx 1 - 7\times 10^{-99}.
\]
Thus the probability of \emph{no} false alerts is on the order of $10^{-99}$,
making the system, for all practical and mathematical purposes,
\emph{certain} to flag at least one innocent person.

The ``one in 30 decillion'' exact-match calculation is mathematically correct but operationally irrelevant. The threshold rule ``at least $m$ of $k$'' transforms the per-person probability from vanishingly small ($10^{-35}$) to operationally significant ($10^{-4}$), and population size amplifies this individual-level vulnerability into system-level certainty of failure. This is not a failure of implementation or algorithm design; it is a mathematical inevitability arising from the combinatorics of threshold detection in high-dimensional attribute spaces.

This contrast between the DNA setting and modern multi-attribute surveillance 
also clarifies why the long-standing DNA database controversy has persisted. 
Classical DNA searches operate in an extremely sparse false-match regime, 
where random matches are so rare that enlarging the database does not 
meaningfully dilute the evidential value of a hit.

Multi-attribute threshold screening, however, operates in a dense false-match 
regime where accumulated coincidences quickly overwhelm the signal. These are 
mathematically distinct operating regimes. Recognizing this distinction 
resolves the apparent conflict in the DNA literature and motivates the 
broader probabilistic limits developed in this paper.

\subsection{Motivation and Background}

Modern surveillance and screening systems collect vast amounts of data about individuals and attempt to identify rare targets of interest, such as potential terrorists, criminals, or fraudsters. These systems compare observed attributes---biometric features, transaction patterns, travel histories, communication metadata---against profiles or watchlists, flagging individuals whose patterns match suspiciously well \cite{lyon2003,oneil2016,schneier2015}.

The fundamental challenge is the \emph{curse of dimensionality in coincidence detection}~\cite{zimek2012}: as monitored attributes grow, innocent individuals increasingly match multiple criteria by chance. When populations are large (millions to billions), even extremely rare coincidences become statistically certain, generating overwhelming false alerts. This parallels challenges in cybersecurity \cite{axelsson2000,lippmann2000, cretu2008, sommer2010}, where even highly accurate intrusion detectors are swamped by false alarms at scale, and in open-set recognition \cite{scheirer2013}, where classifiers must reliably reject vast numbers of unfamiliar inputs.

\subsection{Related Work and Connections}

Our analysis builds on classical probability theory (Poisson approximation, 
Chernoff bounds, concentration inequalities) applied to modern surveillance 
contexts, situating this work within the framework of large-scale 
inference~\cite{efron2010}. The system-level bound $\Pr(\text{false alert}) \approx 1 - (1-q)^n$ 
corresponds to family-wise error rate (FWER) control in multiple 
testing~\cite{benjamini1995,storey2002,benjaminiyekutieli2001}, though our large-deviation analysis 
shows when such control becomes infeasible.

The base-rate analysis resolves the Stockmarr--Balding 
debate~\cite{stockmarr1999,balding2002} in forensic DNA statistics by 
identifying distinct asymptotic regimes. Connections to anomaly 
detection~\cite{axelsson2000,sommer2010}, open-set recognition~\cite{scheirer2013}, 
and algorithmic fairness~\cite{hardt2016,chouldechova2017,corbettdavies2017,
barocas2019,mehrabi2021} appear throughout: false-positive explosion is a fundamental 
challenge across all screening domains.

\subsection{Main Contributions}

This paper develops a rigorous probabilistic framework for analyzing such screening systems. Using classical tools (Poisson approximations, Chernoff bounds, concentration inequalities, and Bayesian inference), we characterize fundamental limits on system reliability. Our analysis reveals that the single-database DNA controversy and multi-attribute surveillance failure are manifestations of the same mathematical phenomenon, differing only in which asymptotic regime they occupy.

Our main contributions include:

\begin{enumerate}
\item \textbf{Sharp critical population bounds:} We derive two-sided bounds on both per-person and system-level false alert probabilities, showing the critical population grows as $n_{\text{crit}} \asymp \sqrt{\lambda}\exp(\lambda D(c\|1))$ with explicit $\sqrt{\lambda}$ corrections (Theorem~\ref{thm:critical_pop}). The rate function $D(c\|1) = c\log c - c + 1$ governs the exponential scaling.

\item \textbf{Finite system lifetimes:} Under exponential data growth $k(t) = k_0\gamma^t$, systems fail at time $T^* \approx \frac{1}{\log\gamma}\log(m/k_0p)$ when $\lambda(t)$ reaches threshold $m$, with population corrections secondary (Theorem~\ref{thm:system_lifetime}). 

\item \textbf{Unified Bayesian-frequentist view:} Both regimes exhibit the same exponential scaling $\exp(\lambda D(c\|1))$, but Bayesian actionability requires $n \lesssim \frac{1-\alpha}{\alpha}\,rs\sqrt{\lambda}e^{\lambda D}$ (for desired PPV $\ge \alpha$) while frequentist reliability allows $n \lesssim \sqrt{\lambda}e^{\lambda D}$ (Section~\ref{sec:bayesian}). Posterior probabilities collapse when $nq \gg rs$, resolving the Stockmarr-Balding controversy as a regime transition.

\item \textbf{Structural bias through group dominance:} Differential surveillance exposure creates exponential outcome disparities that cannot be eliminated through threshold adjustments (Section~\ref{sec:heterogeneous}), formalizing ``structural bias'' arguments from the algorithmic fairness literature.

\item \textbf{Effective dimensionality under correlation:} Spatiotemporal 
dependencies reduce effective attribute counts to 
$k_{\mathrm{eff}} \approx A/(2\pi\xi^2)$ (spatial) or 
$k_{\mathrm{eff}} \approx k/(2\tau)$ (temporal), connecting surveillance 
analysis to multiple-testing corrections in genomics and neuroimaging 
(Section~\ref{sec:spatialtemporal}; technical details in 
Appendix~\ref{app:correlation}).
\end{enumerate}

All main results are proved using elementary probability, with full technical
details collected in Appendix~\ref{app:proofs}. Appendix~\ref{app:empirical}
presents illustrative examples based on public datasets (UCI Adult Census and
Chicago Crime) that show how the phenomena described here arise with real data.

\medskip\noindent
\textbf{Notation.} The symbols above are defined precisely in later sections: 
$D(c\|1)$ is the Poisson rate function (Lemma~\ref{lem:poisson_upper}); 
$c > 1$ is the threshold multiplier with $m = \lceil c\lambda \rceil$ 
(Section~\ref{subsec:assumptions}); $\gamma$ and $k_0$ parameterize 
exponential data growth (Section~\ref{sec:temporal}); $r$, $s$, and $\alpha$ 
are Bayesian parameters (Section~\ref{sec:bayesian}); and $k_{\mathrm{eff}}$, 
$\xi$, $\tau$ quantify correlation effects (Section~\ref{sec:spatialtemporal}). 
A complete notation summary appears in Section~\ref{sec:notation}.

\section{Materials and Methods}

\subsection{The Basic Screening Model}

Consider a population of $n$ individuals, each of whom we collect $k$ types of data (e.g., locations visited, financial transactions, etc.). For each type of data, we also have lists of suspicious activity, for example, locations of interest where crimes have occurred. In general, we observe $k$ binary attributes (indicators) for each individual, where each attribute has a small probability $p \ll 1$ of matching some profile of interest purely by chance. 

For example, a simple model for the match probability $p$ on the $j$-th data type might be found by calculating the probability that a person's list (size $t$) overlaps the suspicious list (size $s$) on a domain of size $V$ possible distinct items (e.g., $V$ geo-location cells). With uniform sampling without replacement from the domain, the probability of an overlap comes from the zero-intersection hypergeometric probability:
\begin{equation*}
p=\Pr(\text{overlap}\ge1)=1-\frac{\binom{V-t}{s}}{\binom{V}{s}}
=1-\prod_{\ell=0}^{s-1}\left(1-\frac{t}{V-\ell}\right),
\quad \text{for } s \le V-t,
\end{equation*}
with $p=1$ when $s > V-t$ (guaranteed overlap).

More generally, the match probability $p$ may be derived from any domain-appropriate model; the hypergeometric calculation above serves merely as one concrete example. Our analysis holds for any fixed value of $p \ll 1$, regardless of how it is obtained. With this understanding, let $X_i$ denote the number of matching attributes for individual $i$. Under the null hypothesis (individual $i$ is innocent) and the modeling assumptions detailed below, we have:
\begin{equation*}
X_i \sim \text{Binomial}(k, p).
\end{equation*}
For large $k$ and small $p$, we write $\lambda = kp$ for the expected number of chance matches; this parameter will govern all subsequent results.

The screening system flags individual $i$ as suspicious if they match some unsolved crime on $m$ attributes with: $X_i \geq m$ for an integer threshold $m$. Throughout, we consider thresholds of the form $m = \lceil c\lambda \rceil$ for constant $c > 1$, representing detection criteria set above the expected number of chance matches; rounding affects only constant factors. The \emph{per-person false alert probability} is:
\begin{equation*}
q = \Pr(X_i \geq m \mid \text{innocent}).
\end{equation*}

Since we screen $n$ individuals, the \emph{system-level false alert probability}, that is, the probability that at least one innocent person is flagged, is:
\begin{equation*}
\Pr(\text{at least one false alert}) = 1 - (1-q)^n.
\end{equation*}

For small $q$, this is commonly approximated by:
\begin{equation*}
\Pr(\text{at least one false alert}) \approx 1 - e^{-nq}.
\end{equation*}

The system becomes \emph{unreliable} when this probability approaches 1 (that is, when $nq \gtrsim 1$), in which case the system is near certain to flag innocent individuals as suspects. 

\subsection{Modeling Assumptions and Regime of Validity}
\label{subsec:assumptions}

Throughout this paper, we operate under the following assumptions:

\begin{enumerate}
\item \textbf{Large-$k$, small-$p$ regime:} We assume $k \gg 1$ (many attributes), $p \ll 1$ (rare individual matches), with mean $\lambda = kp$ held moderate and fixed as $k \to \infty$. This is the natural Poisson regime for rare-event detection, where the binomial distribution converges to Poisson$(\lambda)$ with total-variation error bounded by $2\sum_{i=1}^k p_i^2$ (see Section~\ref{sec:fundamental}).

\item \textbf{Independence across individuals:} The match indicators for different individuals are statistically independent (Remark~\ref{rem:independence}). In practice, common-mode events (e.g., large gatherings or seasonal shifts) may introduce positive dependence, which would only increase false alert rates beyond our bounds.

\item \textbf{Homogeneous match probability (baseline model):} For analytical tractability, we initially assume all attributes have the same match probability $p$. Section~\ref{sec:heterogeneous} extends to heterogeneous populations with different exposure rates $p_g$ across demographic groups. Remark~\ref{rem:heterogeneous_attributes} addresses heterogeneity across attributes within individuals.

\item \textbf{Binary attributes:} Each attribute either matches ($M = 1$) or doesn't match ($M = 0$) a profile of interest. This models presence/absence at locations, yes/no transaction patterns, or binarized continuous features.

\item \textbf{Fixed threshold screening:} The system flags individuals with $m$ or more matches, where $m$ is predetermined. We primarily analyze $m = c\lambda$ for constant $c > 1$, representing thresholds scaled to exceed the expected number of chance matches under the null hypothesis.

\item \textbf{Null hypothesis analysis:} We analyze false alerts under the assumption that all $n$ individuals are innocent (the null hypothesis). Detection of actual targets requires separate treatment via signal detection theory and receiver operating characteristic (ROC) analysis \cite{fawcett2006}, which is not addressed here. Our results characterize the false positive rate; in practice, system designers must balance this against detection power for true targets.
\end{enumerate}

These assumptions define the scope of our analysis. The core results 
(Sections~\ref{sec:fundamental} through~\ref{sec:temporal}) establish 
fundamental limits under these idealized conditions. Later sections extend 
the framework: Section~\ref{sec:heterogeneous} analyzes heterogeneous 
populations with differential exposure, and 
Section~\ref{sec:spatialtemporal} summarizes how spatiotemporal correlation 
reduces effective dimensionality (with technical details in 
Appendix~\ref{app:correlation}).

The large-$k$, small-$p$ regime is particularly natural for modern surveillance
systems, where technological advances enable collection of thousands of
attributes per person (GPS pings, transaction records, communication events).
In our model, each attribute contributes a \emph{binary match indicator}
(\emph{match} or \emph{no match}) of whether it aligns with a profile of
interest, and each such match remains rare. For instance, with $k = 1000$ location cells and $p = 0.005$ (half-percent chance of presence at any flagged location), we have $\lambda = 5$ expected false matches, a moderate value that falls squarely within our analytical framework.

\section{Results}

\subsection{Fundamental Probabilistic Limits}
\label{sec:fundamental}

\subsubsection{The Poisson Approximation}
\label{subsec:poisson_approx}

When $k$ is large, $p$ is small, and $\lambda = kp$ is \emph{fixed}, the
$\mathrm{Binomial}(k,p)$ distribution can be approximated by $\mathrm{Poisson}(\lambda)$ \cite{feller1968}.
Specifically, if $X \sim \mathrm{Binomial}(k,p)$, then
\begin{equation*}
\Pr(X = j) \approx e^{-\lambda}\frac{\lambda^j}{j!}, \qquad \lambda = kp.
\end{equation*}
Moreover, by Le Cam's inequality \cite{barbour1992},
\begin{equation*}
d_{\mathrm{TV}}\!\left(\mathrm{Binomial}(k,p),\mathrm{Poisson}(\lambda)\right) \le 2\sum_{i=1}^k p_i^2.
\end{equation*}
For the homogeneous case with $p_i = p$ for all $i$, this gives $d_{\mathrm{TV}} \le 2kp^2 = 2\lambda p$. When $p$ is small and $\lambda = kp$ is moderate, this bound is $O(p)$, making the Poisson approximation numerically safe for practical applications.
Throughout, we assume independence across individuals; common-mode shocks would require extensions via positively associated random variables (see Remark~\ref{rem:independence}).

Under this approximation, the per-person false alert probability (tail) is
\begin{equation*}
q \;=\; \Pr(X \ge m) \;\approx\; \sum_{j=m}^{\infty} e^{-\lambda}\frac{\lambda^j}{j!}
\;=\; 1 - \sum_{j=0}^{m-1} e^{-\lambda}\frac{\lambda^j}{j!}.
\end{equation*}

\begin{remark}[Heterogeneous Attribute Probabilities]
\label{rem:heterogeneous_attributes}
For analytical tractability we assume a homogeneous match probability $p$ across
attributes, but this assumption can be relaxed. Let the attribute-level match
probabilities be $p_1,\ldots,p_k$, and define
\[
    X = \sum_{i=1}^k \mathrm{Bernoulli}(p_i), 
    \qquad
    \lambda = \sum_{i=1}^k p_i .
\]
By Le~Cam's theorem \cite{barbour1992}, the total variation distance between the
distribution of $X$ and a Poisson distribution with mean $\lambda$ satisfies
\[
    d_{\mathrm{TV}}\!\left( X, \mathrm{Poisson}(\lambda) \right)
    \le 2 \sum_{i=1}^k p_i^2.
\]
Thus the Poisson approximation remains accurate whenever all $p_i$ lie in the
small-probability regime, even if the $p_i$ differ substantially from one
another. In this heterogeneous setting, the distribution of $X$, and therefore
all results in this section, depends on the collection $\{p_i\}$ only through
the aggregate rate~$\lambda$, up to an approximation error of order
$\sum_i p_i^2$.

Group-level heterogeneity (e.g., differing $p_g$ across demographic groups) is
treated separately in Section~\ref{sec:heterogeneous}.
\end{remark}

\begin{remark}[Independence Across Individuals]\label{rem:independence}
Since we screen $n$ individuals independently, the probability that at least one innocent person is flagged is:
   \begin{equation*}
   \Pr(\text{at least one false alert}) = 1 - (1-q)^n.
   \end{equation*}
In real deployments, common-mode events may introduce positive dependence 
among individuals. Positive association means 
$\Pr(X_i \geq x_i \text{ for all } i) \geq \prod_i \Pr(X_i \geq x_i)$: 
variables are more likely to be jointly large than under independence. 
This increases the system-level false-alert probability beyond our bounds.

Our analysis assumes independence. Extensions to positively associated 
arrays via Chen--Stein or Stein's method~\cite{barbour2005} are possible 
but omitted here. The key implication is that positive dependence increases 
tail probabilities, so our independence-based bounds \emph{underestimate} 
true false alert rates when correlation is present. Thus, the independence-based bounds are conservative.
\end{remark}

\subsubsection{Tail Bounds for the Poisson Distribution}

To understand how $q$ depends on $m$ and $\lambda$, we use Chernoff bounds and related concentration inequalities \cite{hoeffding1963,mitzenmacher2005} for Poisson random variables.

\begin{lemma}[Poisson Upper Tail]
\label{lem:poisson_upper}
Let $Y \sim \mathrm{Poisson}(\lambda)$ and suppose $m>\lambda$. Then
\begin{equation}
\Pr(Y \ge m) \;\le\; \left(\frac{e\lambda}{m}\right)^{\!m} e^{-\lambda}
\;=\; \exp\!\Big(-\lambda\,D\!\left(\frac{m}{\lambda}\middle\|1\right)\Big),
\label{eq:poisson_upper_bound}
\end{equation}
where $D(\alpha\|1) = \alpha\log \alpha - \alpha + 1$ is the Poisson rate function.
\end{lemma}

\begin{remark}[Rate Function Properties]
The function $D(\alpha\|1)$ is also called the Cram\'er rate function or 
Cram\'er transform of the Poisson distribution with unit mean. This differs 
from the binary Kullback--Leibler divergence 
$D_{\mathrm{KL}}(p\|q) = p\log(p/q) + (1-p)\log((1-p)/(1-q))$ used for 
Bernoulli random variables.

The large-deviation rate is asymptotically exact for $m/\lambda > 1$: the 
tail probability satisfies 
$q \asymp \lambda^{-1/2}e^{-\lambda D(m/\lambda\|1)}$, with factorial 
discreteness contributing only the $\lambda^{-1/2}$ subexponential prefactor. 
We write $D(c\|1)$ or simply $D$ for brevity when the argument is clear.
\end{remark}

\begin{proof}
For any $t>0$,
\begin{align*}
\Pr(Y \ge m)
&= \Pr\!\big(e^{tY} \ge e^{tm}\big)
\le \frac{\mathbb{E}[e^{tY}]}{e^{tm}}
= \exp\!\big(\lambda(e^t-1) - tm\big).
\end{align*}
Optimizing by setting $\frac{d}{dt}[\lambda (e^t -1) - tm]=0$ gives $e^{t^*}=m/\lambda$,
which yields the stated bound.
\end{proof}

\subsubsection{Critical Population Size}

We now establish the fundamental scale of population size for reliable screening when the decision threshold scales as $m=c\lambda$ with $c>1$.

\begin{theorem}[Critical Population Scale]
\label{thm:critical_pop}
Consider a screening system with $k$ binary attributes, per-attribute match probability $p$, threshold $m=\ceil{c\lambda}$ for some $c>1$, population size $n$, and $\lambda=kp$.
Let $D(c\|1)=c\log c - c + 1$ (the Poisson rate function from Lemma~\ref{lem:poisson_upper}).
Then the per-person false alert probability $q$ satisfies the following asymptotic bounds:
\begin{equation}
\frac{1}{\sqrt{2\pi c\lambda}}\,\exp\!\big(-\lambda D(c\|1) - \tfrac{1}{12c\lambda} \big)
\;\le\;
q
\;\le\;
\exp\!\big(-\lambda D(c\|1)\big),
\label{eq:poisson-tail-two-sided}
\end{equation}
which hold for $\lambda \ge 1$ and $c > 1$ with the stated constants; for smaller $\lambda$ the same exponential form holds with different absolute constants.
The system-level false alert probability obeys
\begin{equation}
1 - \exp\!\Big(
  -\,\frac{n}{\sqrt{2\pi c\lambda}}\,
     e^{-\lambda D(c\|1) - \tfrac{1}{12c\lambda}}
     \Big)
\;\le\;
\Pr(\text{false alert})
\;\le\;
1 - \exp\!\Big(
  -\,\frac{n\,e^{-\lambda D(c\|1)}}{1 - e^{-\lambda D(c\|1)}}
  \Big).
\label{eq:system-bounds}
\end{equation}
(For $e^{-\lambda D} \le 1/2$, the upper bound denominator is within a factor of 2 in the exponent, providing tight control.)
In particular, the system becomes unreliable once
\begin{equation*}
n \;\gtrsim\; \sqrt{\lambda}\;\exp\!\big(\lambda D(c\|1)\big)
\qquad\text{(up to constant factors),}
\end{equation*}
and the critical population scale is
\begin{equation}
n_{\mathrm{crit}} \asymp \exp\!\big(\lambda D(c\|1) \big)=\exp\!\big(kp\,D(c\|1)\big),
\label{eq:ncrit_main}
\end{equation}
ignoring subexponential $\sqrt{\lambda}$ corrections.
Here we use the notation $f \asymp g$ to denote equality up to multiplicative constants independent of the main asymptotic parameters ($\lambda$, $k$, $n$).
\end{theorem}

\begin{proof}[Proof sketch]
The upper bound follows directly from Lemma~\ref{lem:poisson_upper} with 
$m = c\lambda$, giving $q \le \exp(-\lambda D(c\|1))$. The lower bound uses 
Robbins' refinement of Stirling's formula, which shows that 
$\Pr(Y = m)$ has the same exponential rate $D(c\|1)$ up to a subexponential 
$\sqrt{\lambda}$ factor. Because $q = \Pr(Y \ge m) \ge \Pr(Y = m)$, these 
bounds yield matching exponential rates for $q$.

The system-level inequalities \eqref{eq:system-bounds} follow from 
$e^{-nx/(1-x)} \le (1-x)^n \le e^{-nx}$ for $x \in (0,1)$. The critical 
population scale arises when $nq \sim 1$, giving 
$n_{\mathrm{crit}} \approx \sqrt{\lambda}\, e^{\lambda D(c\|1)}$, and hence, 
up to subexponential factors, $n_{\mathrm{crit}} \asymp e^{\lambda D(c\|1)}$ 
as stated. Full derivations appear in Appendix~\ref{app:proofs}.
\end{proof}

\subsubsection{Sharp Threshold Behavior}

The relationship between population size and system reliability exhibits an increasingly sharp threshold as the number of attributes grows.

\begin{proposition}[Threshold Sharpness]
\label{prop:sharp_threshold}
Consider a screening system with $\lambda = kp$, threshold $m = c\lambda$ for fixed $c > 1$, and population scaled as $n = \sqrt{\lambda} \cdot e^{\alpha \lambda D(c\|1)}$ for $\alpha > 0$. Then the system-level false alert probability satisfies:
\begin{equation}
\lim_{\lambda \to \infty} \Pr(\text{false alert}) = 
\begin{cases}
0 & \text{if } \alpha < 1, \\
1 & \text{if } \alpha > 1.
\end{cases}
\label{eq:sharp_threshold}
\end{equation}

The transition occurs in a window $\Delta\alpha \asymp 1/(\lambda D(c\|1)) \to 0$, becoming arbitrarily sharp as $\lambda$ increases; this follows from standard large-deviation theory for sums of i.i.d.\ rare events \cite{dembo2010}. (Note: $\alpha$ here denotes the population scaling exponent, distinct from the posterior probability threshold $\alpha$ in Section~\ref{sec:bayesian}.)
\end{proposition}

\begin{proof}[Proof sketch]
From Theorem~\ref{thm:critical_pop}, $\log q = -\lambda D(c\|1) + O(\log\lambda)$. 
With $n = \sqrt{\lambda}\,e^{\alpha\lambda D(c\|1)}$, we obtain
\[
\log(nq) = (\alpha-1)\lambda D(c\|1) + O(\log\lambda).
\]
For $\alpha < 1$, $nq \to 0$ exponentially; for $\alpha > 1$, $nq \to \infty$ 
exponentially. The transition width $\Delta\alpha = O(1/(\lambda D)) \to 0$. 
Full details appear in Appendix~\ref{app:proof_sharp_threshold}.
\end{proof}

\begin{remark}
This sharp threshold implies no gradual degradation: systems operating near the 
critical line $\log n = \lambda D(c\|1)$ transition from reliable to unreliable 
essentially instantaneously as $\lambda$ increases (Figure~\ref{fig:phase_transitions}). 
Combined with exponential data growth $\lambda(t) = \lambda_0 \gamma^t$
with 
$\lambda_0 = k_0 p$, every system inevitably crosses this threshold at a finite time
\[
T^* \sim \frac{\log(m/\lambda_0)}{\log \gamma}
= \frac{1}{\log \gamma} \log\!\left(\frac{m}{k_0 p}\right)
\]
(Theorem~\ref{thm:system_lifetime}), making failure predictable rather than merely 
possible.
\end{remark}

\subsubsection{Numerical Validation}

Monte Carlo simulations (5000 runs per data point) validate the Poisson 
approximation and critical population predictions as $k$ increases. 
Figure~\ref{fig:phase_transitions}(a) compares simulated false-alert rates 
(orange points) to theoretical predictions, showing mean absolute error of 
0.0015 across the critical transition region. The sharp threshold behavior 
predicted by Proposition~\ref{prop:sharp_threshold} is clearly evident in 
the simulation data.

The accuracy of the Poisson approximation across varying $(k,p)$ combinations 
follows from Le Cam's inequality (Section~\ref{subsec:poisson_approx}). In the 
homogeneous case $p_i = p$ used in this section, Le Cam's bound simplifies to
\[
    d_{\mathrm{TV}} \le 2 k p^2 = 2 \lambda p,
\]
where $\lambda = kp$. Table~\ref{tab:poisson_error} compares this theoretical 
bound against the actual total variation distance between Binomial$(k,p)$ and 
Poisson$(\lambda)$ distributions across parameter combinations spanning the 
operating regime of interest.

\begin{table}[htbp]
\caption{Poisson approximation error: theoretical bound versus actual total 
variation distance for representative $(k,p)$ combinations.}
\label{tab:poisson_error}
\centering
\begin{tabular}{ccccc}
\toprule
$k$ & $p$ & $\lambda = kp$ & Le Cam Bound & Actual $d_{\mathrm{TV}}$ \\
\midrule
50    & 0.10   & 5.0  & 1.000  & 0.026 \\
100   & 0.05   & 5.0  & 0.500  & 0.013 \\
500   & 0.01   & 5.0  & 0.100  & 0.0025 \\
1000  & 0.005  & 5.0  & 0.050  & 0.0012 \\
5000  & 0.001  & 5.0  & 0.010  & 0.0002 \\
\midrule
100   & 0.10   & 10.0 & 2.000  & 0.026 \\
1000  & 0.01   & 10.0 & 0.200  & 0.0025 \\
\bottomrule
\end{tabular}
\end{table}

The Le Cam bound is conservative by a factor of 17--40 in these regimes; actual 
approximation error is substantially smaller than the bound suggests. For the 
primary parameters used throughout this paper ($k = 1000$, $p = 0.005$, 
$\lambda = 5$), the actual total variation distance is approximately $0.001$, 
indicating that the Poisson approximation introduces negligible error relative 
to the phase transition effects we characterize. Even in less favorable regimes 
($k = 50$, $p = 0.1$), the approximation error remains below $0.03$, confirming 
that the Poisson model is appropriate across the full range of parameters 
considered.

Additional numerical illustrations across a broader range of parameters, as
well as examples based on real-world datasets, are given in
Appendix~\ref{app:empirical}.

\medskip\noindent\textbf{Key Takeaway.}
The per-person false alert probability decays as 
$ q \sim e^{-\lambda D(c\|1)} $, and the system becomes unreliable once the 
population size exceeds 
$ n_{\mathrm{crit}} \sim \sqrt{\lambda}\, e^{\lambda D(c\|1)} $. 
This behavior produces a sharp phase transition: systems move abruptly from 
reliable to unreliable operation, and the sharpness of the transition increases 
monotonically with~$\lambda$.

\subsection{Temporal Dynamics: Finite System Lifetimes}
\label{sec:temporal}

Real surveillance systems accumulate data over time. As more attributes are collected, 
the false alert problem intensifies until the system becomes unreliable. We now quantify 
how long a screening system can operate before this inevitable failure occurs. As 
Figure~\ref{fig:phase_transitions}(c) demonstrates, systems with exponential data growth 
exhibit sharp temporal transitions from reliable to unreliable operation at predictable 
times.

\begin{remark}[Change of Threshold Regime]
In the preceding analysis (Section~\ref{sec:fundamental}), we studied thresholds of the form $m = c\lambda$ with $c > 1$ constant, where both $m$ and $\lambda$ scale together. In this section, we adopt a different perspective: we fix the threshold $m$ at deployment time while allowing $\lambda(t) = k(t) \cdot p$ to grow over time as data accumulates. This models realistic system operation, where detection thresholds are predetermined policy parameters that remain constant even as surveillance capabilities expand. The critical behavior occurs when the growing $\lambda(t)$ crosses the fixed threshold $m$, at which point false alerts become inevitable.
\end{remark}

\subsubsection{Exponential Data Growth}

Modern data collection exhibits exponential growth \cite{hilbert2011,reinsel2018}. We model this as:
\begin{equation}
k(t) = k_0 \gamma^t,
\label{eq:data_growth}
\end{equation}
where $k_0$ is the initial number of monitored attributes at time $t=0$, and $\gamma > 1$ is the growth factor per time unit (e.g., annual growth rate).

Since each attribute has probability $p$ of matching by chance for an innocent individual, the expected number of false matches grows exponentially:
\begin{equation}
\lambda(t) = k(t) \cdot p = k_0 p \cdot \gamma^t.
\label{eq:lambda_growth}
\end{equation}

\subsubsection{Critical Time Derivation}

The system flags an individual when they have $m$ or more matching attributes. At time $t$, the per-person false alert probability is:
\begin{equation*}
q(t) = \Pr(\text{Poisson}(\lambda(t)) \geq m).
\end{equation*}

With population size $n$, the probability of at least one false alert is approximately:
\begin{equation*}
\Pr(\text{system false alert at time } t) \approx 1 - e^{-n \cdot q(t)}.
\end{equation*}

The system becomes unreliable when $n \cdot q(t) \gtrsim 1$.

\begin{theorem}[System Lifetime Under Exponential Growth]
\label{thm:system_lifetime}
Consider a screening system with exponential data growth $k(t) = k_0\gamma^t$ (where $\gamma > 1$), per-attribute match probability $p$, fixed threshold $m$, and population size $n$.

For populations satisfying $n \gtrsim \sqrt{m} \exp(m/2)$ (such that system failure occurs near $\lambda(T^*) \approx m$ rather than deep in the Poisson tail where $\lambda \ll m$), the system becomes unreliable at time
\begin{equation}
T^* \approx \frac{1}{\log \gamma} \log\left(\frac{m}{k_0 p}\right),
\label{eq:critical_time_main}
\end{equation}
where the critical time is characterized by $\lambda(T^*) \approx m$.

The population size $n$ introduces a correction of order $\displaystyle \frac{1}{\log \gamma}\sqrt{\frac{\log n}{m}}$ to the time to failure (equivalently, a shift of order $\sqrt{2m\log n}$ in $\lambda$), up to logarithmic factors in $m$ and $n$. This is typically weak relative to the effects of $\gamma$ and $m$.
\end{theorem}

\begin{proof}[Proof sketch]
The system fails when the expected number of false alerts reaches order one: $n \cdot q(T^*) \sim 1$. 
For $\lambda < m$, the Poisson tail probability $q(\lambda) = \Pr(\text{Poisson}(\lambda) \ge m)$ is exponentially suppressed by the factor $\exp(-\lambda D(m/\lambda \| 1))$. As $\lambda$ increases from well below $m$ toward $m$, the rate function $D(m/\lambda \| 1)$ decreases to zero, causing $q(\lambda)$ to increase sharply from nearly zero to order one. When $\lambda > m$, the tail probability increases rapidly toward 1. The transition occurs sharply near $\lambda = m$ due to concentration of the Poisson distribution.

For $\lambda$ near $m$, we use the normal approximation Poisson($\lambda$) $\approx N(\lambda, \lambda)$, which yields $q(\lambda) = \Pr(X \geq m) \approx 1 - \Phi((m-\lambda)/\sqrt{\lambda})$. For $\lambda \lesssim m$, the Gaussian tail asymptotics (Mills' ratio: $1-\Phi(x) \sim \phi(x)/x$ as $x \to \infty$) give the more precise form:
\begin{equation*}
q(\lambda) \sim \frac{\sqrt{\lambda}}{(m-\lambda)\sqrt{2\pi}}\,e^{-(m-\lambda)^2/(2\lambda)}
\qquad (\lambda < m).
\end{equation*}
Setting $nq(T^*) \sim 1$ gives $\lambda(T^*) \approx m - \sqrt{2m\log n}$, which is a small correction to $\lambda(T^*) = m$ for typical parameters. Translating through $\lambda(t)=k_0 p\,\gamma^t$ yields a time shift of order $\frac{1}{\log \gamma}\sqrt{\frac{\log n}{m}}$. Substituting into $\lambda(t) = k_0 p \gamma^t$ yields \eqref{eq:critical_time_main} as the leading-order term.
\end{proof}

\subsubsection{Key Insights}

Equation \eqref{eq:critical_time_main} reveals several fundamental properties:

\begin{enumerate}
\item \textbf{Finite lifetime is inevitable:} 
For any fixed threshold $m$ and exponential growth $\gamma > 1$, we have $T^* < \infty$.

Even if one attempts to preserve reliability by scaling the threshold with the data ($m \propto \lambda(t)$), doing so would require the monitored population to grow as 
$n_{\mathrm{crit}}(t) \sim \exp(C \gamma^{t})$ for some $C>0$---a doubly-exponential rate that far exceeds any realistic population dynamics.

Because real surveillance systems use fixed thresholds and monitor approximately fixed populations, temporal failure becomes inevitable rather than merely possible. This mathematical structure parallels the ``double birthday paradox'' in our earlier work~\cite{pollanen2024}, though it arises here in a distinct screening context.

\item \textbf{Logarithmic dependence on threshold:} Doubling the threshold $m$ adds only $(\log 2)/\log\gamma$ time units. For annual doubling ($\gamma = 2$), this is exactly 1 year. Even increasing $m$ tenfold extends lifetime by only $\log_{2}(10) \approx 3.3$ years.

\item \textbf{Inverse dependence on growth rate:} The factor $1/\log\gamma$ means faster data growth dramatically reduces system lifetime. Increasing $\gamma$ from 1.5 to 2 (from 50\% to 100\% annual growth) roughly halves the operational lifetime.

\item \textbf{Weak population dependence:} While larger populations cause slightly earlier failure, this effect is logarithmic in $n$ and secondary to the exponential effects of $\gamma$ and $m$. The system lifetime is primarily determined by data growth dynamics, not population size.
\end{enumerate}

The temporal phase transition in Figure~\ref{fig:phase_transitions}(c) illustrates 
these dynamics: a system with 50\% annual growth ($\gamma = 1.5$) operates reliably 
for approximately 4 years before crossing the critical threshold, after which false 
alerts become statistically inevitable.

\subsubsection{Quantitative Examples}

\begin{example}[Border Security System]
\label{ex:temporal_scaling}
Consider a border security system with initial attributes $k_0 = 100$, annual growth rate $\gamma = 1.5$ (50\%), match probability $p = 0.01$, and threshold $m = 5$. The critical time is:
\begin{equation*}
T^* = \frac{\log(5/(100 \times 0.01))}{\log 1.5} = \frac{\log 5}{\log 1.5} \approx 4.0 \text{ years}.
\end{equation*}

The lifetime sensitivity reveals fundamental constraints. Increasing $m$ from 3 to 10:
\begin{align*}
T^*(m=3) &= \frac{\log(3/(100 \times 0.01))}{\log 1.5} = \frac{\log 3}{\log 1.5} \approx 2.7 \text{ years},\\
T^*(m=10) &= \frac{\log(10/(100 \times 0.01))}{\log 1.5} = \frac{\log 10}{\log 1.5} \approx 5.7 \text{ years}.
\end{align*}
Increasing the threshold from 3 to 10 extends lifetime from 2.7 to only 5.7 years (logarithmic benefit). In contrast, doubling $\gamma$ from 1.5 to 2.0:
\begin{equation*}
T^*(m=5, \gamma=2.0) = \frac{\log 5}{\log 2} \approx 2.3 \text{ years},
\end{equation*}
cuts lifetime nearly in half (from 4.0 to 2.3 years).
Growth rate dominates system longevity; threshold adjustments provide minimal protection.
\end{example}

\subsubsection{Practical Implications}

\begin{enumerate}
\item \textbf{Predictable obsolescence:} Every system has a calculable expiration date $T^* \sim \frac{1}{\log\gamma}\log(m/k_0 p)$. Monitoring $\lambda(t)/m$ provides early warning.

\item \textbf{Growth rate dominates:} Lifetime scales as $1/\log\gamma$; doubling $\gamma$ halves operational lifetime. Threshold adjustments provide only logarithmic benefit.
\end{enumerate}

\begin{remark}[Non-Exponential Data Growth]
\label{rem:non_exponential}
The analysis above focuses on exponential growth, but alternative data-growth 
models lead to qualitatively different system lifetimes. For simplicity, we 
identify the onset of system unreliability with the heuristic condition 
$\lambda(t) = p\,k(t) \approx m$, which corresponds to the point at which 
false-alert probabilities increase sharply. We assume $m > p k_0$ so that 
this regime is reached, if at all, at some $T^* > 0$.

\begin{itemize}
\item \textbf{Polynomial growth} $k(t) = k_0 + \gamma t^\alpha$ with $\alpha \ge 1$:
Solving $p(k_0 + \gamma T^{*\alpha}) \approx m$ gives
\[
T^* \approx \left( \frac{m/p - k_0}{\gamma} \right)^{1/\alpha},
\]
so $T^* \propto m^{1/\alpha}$ grows linearly when $\alpha = 1$ and
sublinearly when $\alpha > 1$. Larger values of $\alpha$ correspond to
faster growth of $k(t)$, leading to shorter system lifetimes. Moreover,
because $1/\alpha$ decreases with $\alpha$, increases in $m$ yield
progressively smaller gains in $T^*$, further reducing threshold
protection. Nonetheless, polynomial growth still provides significantly
more protection than exponential growth, where $T^*$ increases only
logarithmically with $m$.

\item \textbf{Logarithmic growth} $k(t) = k_0 + \beta\log(1 + \gamma t)$: 
Here
\[
T^* \approx \frac{\exp\!\big((m/p - k_0)/\beta\big) - 1}{\gamma},
\]
so the time required for $\lambda(t)$ to reach $m$ grows exponentially in the 
threshold $m$. Because $k(t)$ itself grows only logarithmically in $t$, this 
time can be extremely large, and practical data collection may saturate before 
a critical regime is ever reached.

\item \textbf{Bounded growth} $k(t) \to k_{\max}$: 
If the maximal attainable rate $\lambda_{\max} = p\,k_{\max}$ satisfies 
$\lambda_{\max} < m$, the system never enters the high-false-alert regime and 
remains below the critical level for all time.
\end{itemize}

Exponential growth is the critical case because it eventually exceeds any 
subexponential threshold-adjustment strategy: if $\lambda(t)$ grows 
exponentially while $m(t)$ increases subexponentially, then 
$\lambda(t) > m(t)$ for all sufficiently large $t$. Empirical studies suggest 
that historical data-collection capacity has grown approximately exponentially 
\cite{hilbert2011}, making this the scenario of greatest relevance for policy 
analysis.
\end{remark}

\begin{figure}[htbp]
\centering
\includegraphics[width=0.95\textwidth]{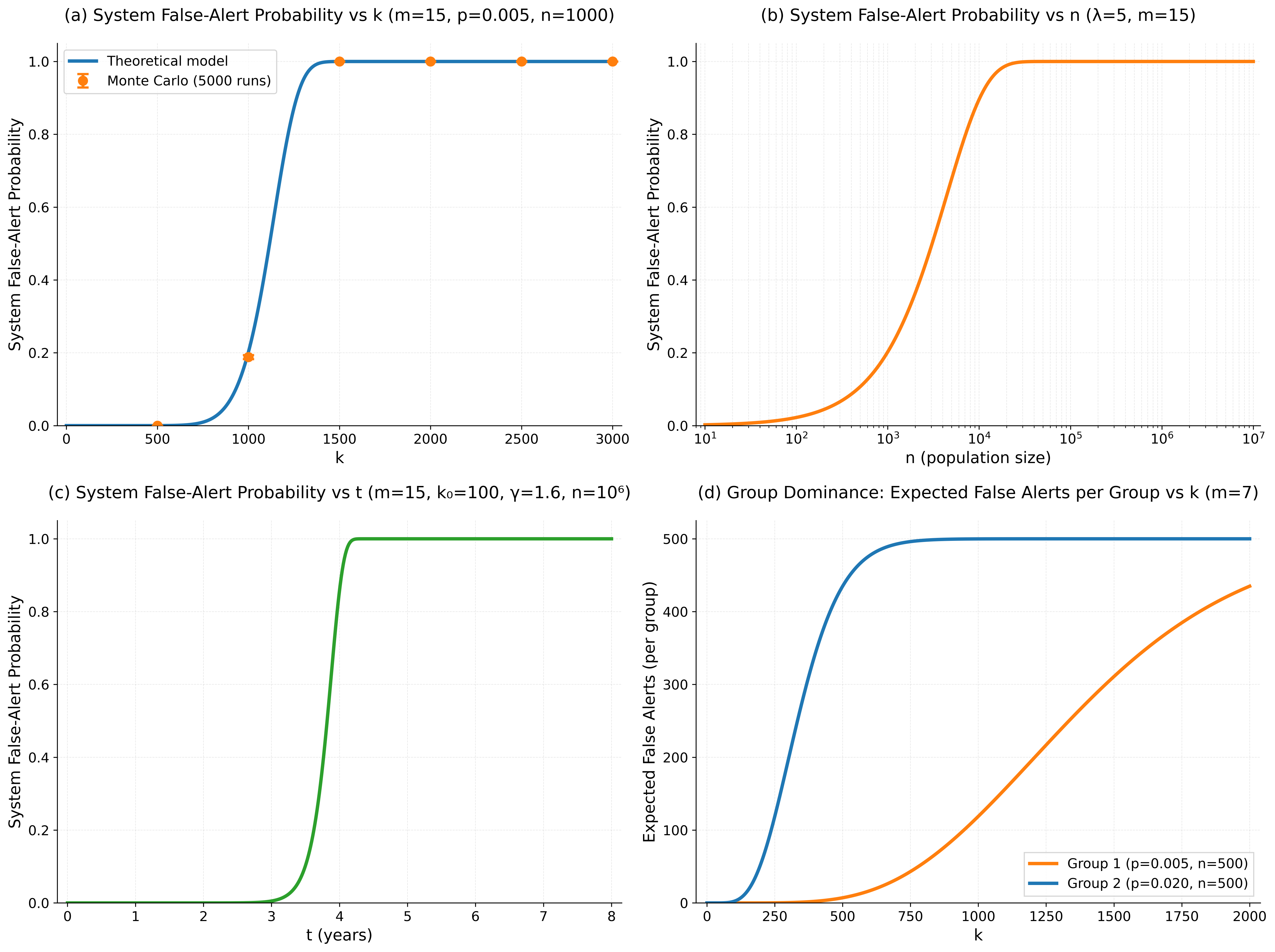}
\caption{\textbf{Phase transitions in surveillance system reliability.} 
Systems exhibit sharp transitions from reliable to unreliable operation across four dimensions. 
\textbf{(a) Attribute growth:} Monte Carlo simulations (orange points, 5000 runs per data point) validate the theoretical predictions (blue curve) with mean absolute error below 0.002. \emph{Takeaway:} The sharp S-curve illustrates a phase transition: systems remain reliable until a critical attribute count is reached, after which reliability collapses rapidly with almost no intermediate zone. 
\textbf{(b) Population scaling:} False-alert probability as a function of population size (log scale) for fixed $\lambda$ and threshold. \emph{Takeaway:} The transition sharpens as $\lambda$ increases, confirming Proposition~\ref{prop:sharp_threshold}. Populations below $n_{\mathrm{crit}} \sim e^{\lambda D}$ are reliable, while those above this scale almost certainly generate false alerts. 
\textbf{(c) Temporal dynamics:} Under exponential data growth with $\gamma = 1.5$ (50\% annual growth), systems transition from reliable to unreliable at a predictable time $T^* \approx 4$ years (Theorem~\ref{thm:system_lifetime}). \emph{Takeaway:} Degradation is abrupt rather than gradual---systems remain functional until they cross a critical time threshold and then fail rapidly. 
\textbf{(d) Group dominance:} Two groups with different exposure rates ($p_1 = 0.005$ vs.\ $p_2 = 0.02$) exhibit markedly different false-alert trajectories. \emph{Takeaway:} The high-exposure group (solid) reaches the failure regime much sooner than the low-exposure group (dashed), illustrating the structural exposure-driven disparities described in Proposition~\ref{prop:exposure_amplification}.}
\label{fig:phase_transitions}
\end{figure}

\medskip\noindent\textbf{Key Takeaway.}
Under exponential data growth, every screening system has a finite, calculable 
lifetime $T^* \approx \frac{1}{\log \gamma}\log(m/k_0 p)$ determined by when 
$\lambda(t)$ reaches the threshold $m$. Threshold adjustments provide only 
logarithmic protection against exponential growth; the growth rate $\gamma$ 
dominates system longevity, making temporal failure inevitable rather than 
merely possible.

\subsection{Heterogeneous Population Structure}
\label{sec:heterogeneous}

Not all individuals have equal representation in surveillance databases. Socioeconomic, geographic, and demographic factors lead to differential exposure rates \cite{lyon2003,eubanks2018}.

\subsubsection{Multi-Group Model}

Partition the population into $G$ groups, such that group $g$ has $n_g$ individuals, with $\sum_{g=1}^{G} n_g = n$, and
each individual in group $g$ has per-attribute match probability $p_g$. Let $\lambda_g = kp_g$ and define the per-group false alert probability:
\begin{equation*}
q_g = \Pr(\text{Poisson}(\lambda_g) \geq m).
\end{equation*}

\begin{proposition}[Heterogeneous System Risk]
\label{prop:hetero_risk}
The system-level false alert probability is:
\begin{equation}
\Pr(\text{false alert}) = 1 - \prod_{g=1}^{G} (1 - q_g)^{n_g} \approx 1 - \exp\left(-\sum_{g=1}^{G} n_g q_g\right).
\label{eq:hetero_system_risk}
\end{equation}
\end{proposition}

\begin{proof}
The events ``individual $i$ in group $g$ generates a false alert'' are independent across individuals. The probability that no one in group $g$ alerts is $(1 - q_g)^{n_g}$. Since groups are disjoint, the probability that no one alerts is the product over groups. The approximation follows from $(1-x)^n \approx e^{-nx}$ for small $x$.
\end{proof}

\subsubsection{Dominance and Disparity}

\begin{theorem}[Group Dominance Effect: Algebraic Decomposition]
\label{thm:dominance}
Let $g^* = \argmax_g \{n_g q_g\}$ be the group contributing the most to system risk. Then:
\begin{equation}
\Pr(\text{false alert}) \approx 1 - e^{-n_{g^*}q_{g^*}} + O\!\left(\Bigl(\sum_{g \neq g^*} n_g q_g\Bigr) e^{-n_{g^*}q_{g^*}}\right).
\label{eq:dominance}
\end{equation}
If $n_{g^*}q_{g^*} \gg \sum_{g \neq g^*} n_g q_g$, then group $g^*$ dominates system behavior. This is an algebraic decomposition showing how system-level risk concentrates in the highest-exposure group.
\end{theorem}

\begin{proof}
From Proposition~\ref{prop:hetero_risk}:
\begin{equation*}
1 - \exp\left(-\sum_{g=1}^{G} n_g q_g\right) = 1 - \exp(-n_{g^*}q_{g^*}) \cdot \exp\left(-\sum_{g \neq g^*} n_g q_g\right).
\end{equation*}
Using $e^{-x} \approx 1 - x$ for small $x$:
\begin{align*}
&\approx 1 - e^{-n_{g^*}q_{g^*}}\left(1 - \sum_{g \neq g^*} n_g q_g\right) \\
&= 1 - e^{-n_{g^*}q_{g^*}} + e^{-n_{g^*}q_{g^*}}\sum_{g \neq g^*} n_g q_g.
\end{align*}
The main term is $1 - e^{-n_{g^*}q_{g^*}}$, representing the contribution from group $g^*$. The correction term is of order $\bigl(\sum_{g \neq g^*} n_g q_g\bigr) e^{-n_{g^*}q_{g^*}}$ and becomes negligible whenever $n_{g^*} q_{g^*} \gg \sum_{g \neq g^*} n_g q_g$.
\end{proof}

\subsubsection{Numerical Example}

\begin{example}[Geographic Disparities]
\label{ex:geographic_disparities}

Consider a city with two neighborhoods:
\begin{itemize}
\item Neighborhood A (low surveillance): $n_A = 100{,}000$, $p_A = 0.005$
\item Neighborhood B (high surveillance): $n_B = 100{,}000$, $p_B = 0.02$
\end{itemize}

With $k = 100$ attributes and threshold $m = 3$:
\begin{align*}
\lambda_A &= 100 \cdot 0.005 = 0.5, \quad q_A = \sum_{j=3}^{\infty} e^{-0.5}\frac{0.5^j}{j!} \approx 0.014, \\
\lambda_B &= 100 \cdot 0.02 = 2, \quad q_B = \sum_{j=3}^{\infty} e^{-2}\frac{2^j}{j!} \approx 0.323.
\end{align*}

Expected number of false alerts:
\begin{align*}
n_A q_A &\approx 100{,}000 \cdot 0.014 = 1{,}400, \\
n_B q_B &\approx 100{,}000 \cdot 0.323 = 32{,}300.
\end{align*}

Despite equal population sizes, Neighborhood B experiences approximately 23 times 
more false alerts. This exponential disparity, visualized in 
Figure~\ref{fig:phase_transitions}(d), demonstrates how differential exposure rates 
create fundamentally unequal outcomes that cannot be remedied through threshold 
adjustments alone. Moreover, by Theorem~\ref{thm:dominance}, since $n_B q_B \gg n_A q_A$, the system-level false alert probability is dominated by Neighborhood B:
\begin{equation*}
\Pr(\text{system false alert}) \approx 1 - e^{-32{,}300} \approx 1.
\end{equation*}

The system is essentially guaranteed to produce false alerts, driven almost entirely by the heavily surveilled neighborhood. This illustrates how heterogeneous exposure not only creates disparate individual-level burdens but also determines aggregate system reliability.
\end{example}

In the UCI Adult Census data, increasing the threshold produces the predicted
divergence in group contributions: at $m = 8$, a 2\% subgroup generates 17\% of
false alerts, while the lowest-exposure group contributes none. See
Appendix~\ref{app:empirical}.

\subsubsection{Exposure Amplification Through Poisson Tails}

\begin{proposition}[Exposure Amplification Through Poisson Tails]
\label{prop:exposure_amplification}
Consider two groups with equal population size ($n_1 = n_2 = n/2$) but different per-attribute match probabilities $p_1 < p_2$. Let $\lambda_i = kp_i$ and suppose both groups are screened with common threshold $m$. Then:
\begin{enumerate}
\item If $\lambda_1 < m \le \lambda_2$, the disparity in expected false 
alerts is exponential:
\begin{equation}
\frac{n_2 q_2}{n_1 q_1} = \frac{q_2}{q_1} \ge \exp(c \cdot m)
\label{eq:exponential_disparity}
\end{equation}
for some constant $c > 0$ depending on $\lambda_1$, for all sufficiently 
large $m$. The disparity grows super-exponentially because $q_1$ decays 
exponentially while $q_2$ remains bounded away from zero.
\item The fraction of false alerts attributable to Group 2 approaches 1 as the exposure gap grows:
\begin{equation*}
\frac{n_2 q_2}{n_1 q_1 + n_2 q_2} \to 1 \quad \text{as } \lambda_2 - m \text{ increases}.
\end{equation*}
\item For fixed threshold $m$ and exposure ratio $\alpha = p_2/p_1 > 1$, as $k$ increases, Group 2 enters the critical regime ($\lambda_2 \approx m$) before Group 1, creating a temporal window of maximum disparity.
\end{enumerate}
\end{proposition}

\begin{proof}[Proof sketch]
(1) For $\lambda_2 \ge m$, we have $q_2 \ge \Pr(\mathrm{Poisson}(m) \ge m) 
\ge 1/2 - O(m^{-1/2})$ by normal approximation, while 
$q_1 \le e^{-\lambda_1 D(m/\lambda_1\|1)}$ is exponentially small 
(Lemma~\ref{lem:poisson_upper}). The ratio grows super-exponentially.
(2) Since $q_2/q_1 \to \infty$, we have $q_2/(q_1 + q_2) \to 1$.
(3) Group~2 reaches $\lambda_2 = m$ at $k = m/p_2 < m/p_1$, so 
disparity is maximal for $k \in (m/p_2, m/p_1)$.
Full details appear in Appendix~\ref{app:proof_exposure}.
\end{proof}

\begin{remark}[Within-Group Individual Heterogeneity]
\label{rem:within_group}
Even within a demographic group $g$, individuals may experience different 
exposure levels due to occupation, neighborhood, behavioral patterns, or 
daily routines, leading to person-specific match probabilities 
$p_g^{(i)}$. A convenient model treats $p_g^{(i)}$ as random with mean 
$\bar{p}_g$, for example $p_g^{(i)} \sim \mathrm{Beta}(\alpha_g,\beta_g)$ with 
$\mathbb{E}[p_g^{(i)}] = \bar{p}_g$. This induces a mixture of Poisson 
distributions for match counts whose variance exceeds the Poisson variance 
of the homogeneous case (classical overdispersion).

The qualitative phenomena we describe persist under within-group 
heterogeneity. As long as the group-level means $\bar{p}_g$ differ, the 
higher-exposure group reaches false-alert thresholds earlier on average. 
Moreover, within-group heterogeneity may actually \emph{increase} disparity: 
overdispersion inflates upper-tail probabilities, making heterogeneous groups 
more likely to produce false alerts than homogeneous groups with the same 
mean exposure. A full mixture analysis is beyond our present scope, but the 
structural exposure-disparity mechanism is robust to this generalization.
\end{remark}

\medskip\noindent\textbf{Key Takeaway.}
Differential surveillance exposure creates exponential disparities in 
false-alert rates through Poisson tail behavior. Because this mechanism arises 
from data collection intensity rather than classifier design, neither threshold 
adjustment nor standard algorithmic fairness interventions can eliminate it. 
The disparity persists under within-group heterogeneity and may even intensify 
due to overdispersion. Achieving genuine outcome parity requires equalizing 
surveillance intensity itself.

\subsection{Effective Dimensionality Under Correlation}
\label{sec:spatialtemporal}

Real surveillance data exhibit spatial and temporal dependencies that 
fundamentally alter tail probabilities. While correlation does not change the 
expected number of matches ($\mathbb{E}[\sum_i X_i] = kp$ regardless of 
dependence), it induces \emph{overdispersion} that inflates tail probabilities 
beyond the independent Poisson approximation.

The key insight is that positive correlation reduces the \emph{effective 
degrees of freedom}. Standard design-effect methodology parameterizes this 
through an effective sample size $k_{\mathrm{eff}}$ satisfying 
$\mathrm{Var}(Y) = k_{\mathrm{eff}} \cdot p(1-p)$. For spatial correlation 
with exponential decay over area $A$ with correlation length $\xi$:
\begin{equation}
k_{\mathrm{eff}} \approx \frac{A}{2\pi\xi^2}.
\label{eq:keff_spatial}
\end{equation}
For temporal correlation with correlation time $\tau$:
\begin{equation}
k_{\mathrm{eff}} \approx \frac{k}{2\tau}.
\label{eq:keff_temporal_simplified}
\end{equation}

This modifies the critical population scale. The exponent 
$\lambda D(c\|1)$ is reduced by the factor $k_{\mathrm{eff}}/k$, yielding:
\begin{equation}
n_{\mathrm{crit}}^{\mathrm{corr}} \lesssim \sqrt{\lambda} 
\exp\!\left(\frac{k_{\mathrm{eff}}}{k} \cdot \lambda D(c\|1)\right).
\label{eq:ncrit_correlation}
\end{equation}

The practical impact is severe. City-scale location monitoring with 
$k = 10{,}000$ cells and correlation length $\xi = 500$m yields 
$k_{\mathrm{eff}} \approx 64$, reducing critical populations by over two 
orders of magnitude. Year-long daily observations ($k = 365$) with 
$\tau = 30$ days yields $k_{\mathrm{eff}} \approx 6$, making reliable 
surveillance of even small groups challenging.

\begin{remark}[Scope: Heuristic Approximations]
\label{rem:heuristic_correlation}
This analysis uses heuristic approximations rather than rigorous 
large-deviation theory. The effective degrees of freedom approach captures 
the dominant effect (variance inflation) but does not constitute a formal 
theorem. The qualitative conclusion that positive correlation accelerates 
false-alert saturation is robust and well established in the literature on 
design effects, effective sample size, and correlation-adjusted multiple 
testing. Full technical details, derivations, and worked examples appear 
in Appendix~\ref{app:correlation}.
\end{remark}

\section{Discussion}

\subsection{Bayesian Posterior Reliability and the Base-Rate Trap}
\label{sec:bayesian}

The preceding analysis focused on frequentist system reliability: the probability 
that at least one innocent individual is flagged. However, practitioners ultimately 
need the \emph{posterior probability} that a flagged person is truly a target. 
This Bayesian perspective reveals an even more stringent constraint: in the 
sparse-target regime (where the expected number of true targets $r$ is small 
relative to population size), posterior reliability degrades once $nq$ becomes 
comparable to $rs$, and collapses when $nq \gg rs$. In this regime, flags become 
epistemically meaningless well before the frequentist transition at $nq \sim 1$. 

\subsubsection{The Bayesian Framework and Positive Predictive Value}

We introduce standard notation from diagnostic testing and forensic statistics \cite{balding2002}:
\begin{align*}
\pi &\equiv \Pr(\text{target}) \quad \text{(base rate)}, \\
s &\equiv \Pr(\text{flag}\mid\text{target}) \quad \text{(sensitivity)}, \\
q &\equiv \Pr(\text{flag}\mid\text{innocent}) \quad \text{(false positive rate)}.
\end{align*}

By Bayes' rule, the \emph{positive predictive value (PPV)} is
\begin{equation}
\label{eq:PPV}
\Pr(\text{target}\mid \text{flag}) 
= \frac{s\,\pi}{s\,\pi + q\,(1-\pi)}.
\end{equation}

In a screened population of size $n$ with $r$ true targets ($\pi = r/n$), the 
expected number of flagged individuals is:
\begin{equation*}
\mathbb{E}[\#\text{ flags}] \approx r\,s + n\,q.
\end{equation*}

Thus, the expected false discovery rate (FDR) is
\begin{equation}
\label{eq:FDR}
\mathrm{FDR} \approx \frac{n\,q}{r\,s + n\,q},
\qquad
\mathrm{PPV} = 1 - \mathrm{FDR} \approx \frac{r\,s}{r\,s + n\,q},
\end{equation}
matching \eqref{eq:PPV} exactly.

\begin{remark}[Sensitivity Dependence on Threshold]
\label{rem:sensitivity_threshold}
The analysis above treats sensitivity $s$ as constant, but in practice 
$s = \Pr(\text{flag} \mid \text{target})$ typically decreases as the 
threshold $m$ increases: stricter criteria miss a larger fraction of true 
targets. A simple parametric model capturing this tradeoff is
\begin{equation}
s(m) = s_{\max} \exp\!\bigl(-\beta (m - \lambda_{\mathrm{signal}})_+\bigr),
\label{eq:s_of_m}
\end{equation}
where $s_{\max}$ is sensitivity at low thresholds, 
$\lambda_{\mathrm{signal}} > \lambda = kp$ is the expected number of 
matching attributes for a true target, $\beta > 0$ controls the decay rate, 
and $(x)_+ = \max(x, 0)$. The exponential form reflects the common 
observation that match-score distributions exhibit approximately exponential 
tails, though any monotone decreasing form would yield the same qualitative 
conclusions.

Substituting \eqref{eq:s_of_m} into the PPV formula \eqref{eq:FDR} reveals 
competing effects as $m$ increases:
\begin{itemize}
\item \textbf{False positives decrease:} for innocent individuals, 
$q(m) = \Pr(Y \ge m)$ with $Y \sim \mathrm{Poisson}(\lambda)$ and 
$\lambda = kp$, and Lemma~\ref{lem:poisson_upper} shows that $q(m)$ 
decays very rapidly once $m > \lambda$ (indeed, faster than any fixed-rate 
exponential in $m$).
\item \textbf{True positives decrease:} $s(m)$ remains near $s_{\max}$ 
for $m \le \lambda_{\mathrm{signal}}$ but then decays exponentially 
for $m > \lambda_{\mathrm{signal}}$.
\end{itemize}

These effects create an \emph{intermediate regime} in which PPV may 
\emph{plateau} or improve more slowly than the constant-$s$ analysis 
suggests, particularly once $m$ exceeds $\lambda_{\mathrm{signal}}$ and 
losses in sensitivity offset some of the gains from reduced false positives. 
Under the model \eqref{eq:s_of_m}, however, Lemma~\ref{lem:poisson_upper} 
implies that $q(m)$ eventually decays much faster than $s(m)$, so 
$q(m)/s(m) \to 0$ and therefore $\mathrm{PPV}(m) \to 1$ as $m \to \infty$. 
Any plateau or dip in PPV can therefore occur only over this intermediate 
range of thresholds, not asymptotically.

Because sensitivity is bounded above by $s_{\max}$, the constant-$s$ model 
used earlier provides an \emph{upper bound} on achievable PPV for any given 
threshold: $\mathrm{PPV}(m)$ is increasing in $s$, and $s(m) \le s_{\max}$ 
for all $m$. Optimal threshold selection ultimately requires specifying both 
the innocent and target match distributions---equivalently, full ROC curve 
analysis \cite{fawcett2006}. Since the signal distribution is 
application-dependent, we do not pursue this direction here.
Of course, driving $m$ to extremely large values also drives the overall flag 
rate toward zero, so the $\mathrm{PPV}(m) \to 1$ limit is primarily of 
theoretical interest; in practice, the operationally relevant regime is the 
intermediate range of thresholds where nontrivial detection rates are maintained.
\end{remark}

\subsubsection{Posterior Reliability and Bayesian Critical Scales}

In the large-deviation setting of Theorem~\ref{thm:critical_pop}, with 
threshold $m = c\lambda$ and $c > 1$, the false positive rate satisfies
\begin{equation*}
q \asymp \kappa(\lambda)\,e^{-\lambda D(c\|1)}, 
\qquad 
\kappa(\lambda) = O(\lambda^{-1/2}),
\end{equation*}
up to subexponential prefactors. Combining this with \eqref{eq:FDR} yields an 
explicit condition for maintaining actionable posterior probabilities.

\begin{remark}[Sparse-Target Regime]
The following proposition applies to the ``needle in a haystack'' 
setting where $r$ (the expected number of true targets) is fixed or grows slowly while population $n$ increases. If instead 
$r \propto n$ (constant prevalence), then the ratio $r/n$ is fixed, and the 
condition for actionable PPV reduces to a bound on $q/s$ relative to 
$\alpha/(1-\alpha)$, independent of $n$. The sparse-target regime is the most 
challenging case for screening systems and is therefore the primary focus of 
this analysis.
\end{remark}

\begin{proposition}[Bayesian Critical Population for Actionable PPV]
\label{prop:bayes_critical}
Fix a desired posterior level $\alpha \in (0,1)$ (e.g., $\alpha=0.9$) and 
sensitivity $s \in (0,1]$. Let $r$ denote the expected number of true targets 
in the population (so the base rate is $\pi=r/n$). In the sparse-target 
regime where $r$ is fixed or grows sublinearly with $n$, the condition
\begin{equation*}
\Pr(\text{target}\mid\text{flag}) \ge \alpha
\end{equation*}
is satisfied whenever
\begin{equation}
\label{eq:bayes_n_condition_exact}
n \le \frac{(1-\alpha)\,r\,s}{\alpha\,q}.
\end{equation}

Under the large-deviation scaling $q \asymp \kappa(\lambda)e^{-\lambda D(c\|1)}$, 
the \emph{Bayesian critical population size} satisfies
\begin{equation}
\label{eq:bayes_critical_scale}
n_{\mathrm{crit}}^{\mathrm{Bayes}}(\alpha,s)
\asymp
\frac{(1-\alpha)\,r\,s}{\alpha}\,\sqrt{\lambda}\,
\exp\!\big(\lambda D(c\|1)\big),
\end{equation}
where $\asymp$ hides subexponential factors absorbed into $\kappa(\lambda)$.
\end{proposition}

\begin{proof}
From \eqref{eq:FDR}, $\mathrm{PPV} \ge \alpha$ iff
$r\,s \ge \alpha(r\,s + n\,q)$.
Rearranging gives
$r\,s (1-\alpha) \ge \alpha\,n\,q$,
which yields \eqref{eq:bayes_n_condition_exact}. Substituting the 
large-deviation scaling for $q$ proves \eqref{eq:bayes_critical_scale}.
\end{proof}

When sensitivity varies with threshold as in Remark~\ref{rem:sensitivity_threshold}, 
the bound \eqref{eq:bayes_critical_scale} should be interpreted as an upper bound 
on achievable Bayesian reliability, since $s(m) \le s_{\max}$ for all $m$.

\begin{remark}[The Bayesian Trap]
Frequentist reliability deteriorates once $nq \sim 1$, when the system is 
likely to produce at least one false alert. Bayesian actionability demands the 
stronger condition $nq \ll rs$. When $rs > 1$, there is an intermediate regime 
where false alerts occur frequently but individual flags retain some evidential 
value. When $rs < 1$ (extremely sparse targets), posterior reliability collapses 
before the system becomes statistically unreliable. In all cases, if $nq \gg rs$, 
posterior probabilities decay toward zero even when individual false positives 
remain rare.
\end{remark}

\subsubsection{Likelihood Ratios and Classical Fallacies}

Analysts often cite tiny false positive rates $q$ or large likelihood ratios
\begin{equation*}
L = \frac{s}{q},
\end{equation*}
and mistakenly infer that $\Pr(\text{target}\mid\text{flag})$ is therefore large. 
This is the classical \emph{prosecutor's fallacy} \cite{balding2002}. Bayes' rule 
shows:
\begin{equation}
\label{eq:posterior_odds}
\frac{\Pr(\text{target}\mid \text{flag})}{\Pr(\text{innocent}\mid \text{flag})}
=
\frac{s}{q} \cdot \frac{\pi}{1-\pi}
=
L \times \text{prior odds}.
\end{equation}

When the base rate $\pi=r/n$ is small, the prior odds can overwhelm any fixed 
likelihood ratio. Even extremely rare false positives ($q \ll 1$) do not guarantee 
high PPV. When $\pi \ll q$, most flagged individuals remain innocent despite 
individually low false positive rates.

\subsubsection{Resolving the DNA Database Controversy}

Forensic statisticians have long debated the evidential value of DNA ``cold hits.'' 
Stockmarr \cite{stockmarr1999} argued that searching databases weakens evidence by inflating 
coincidental match probabilities; Balding \cite{balding1995,balding2002} countered that 
likelihood ratios preserve evidential weight. Our framework resolves this 
apparent contradiction by identifying the relevant asymptotic regime.

Using \eqref{eq:posterior_odds}, the posterior odds after a database match are:
\begin{equation*}
\frac{\Pr(\text{guilty}\mid\text{match})}{\Pr(\text{innocent}\mid\text{match})}
= \frac{s}{q} \cdot \frac{r/n}{1-r/n}.
\end{equation*}

\textbf{The DNA regime.}  
Standard STR genotype profiling yields match probabilities on the order of 
$q \sim 10^{-12}$ or smaller. Even with databases of size $n \sim 10^{6}$, the 
product $nq \sim 10^{-6} \ll 1$ remains far below the critical scale where 
coincidental matches become probable.

In this extreme regime, likelihood ratios dominate the posterior odds. The prior 
odds $r/n$ may be small (perhaps $10^{-6}$ if we have one suspect among a million), 
but the likelihood ratio $s/q \sim 10^{12}$ is so enormous that posterior 
probabilities remain overwhelmingly high. This validates Balding's argument: 
database size does not meaningfully dilute evidential weight when $nq \ll 1$.

\textbf{The surveillance regime.}  
Multi-attribute surveillance systems operate in a fundamentally different regime. 
With $\lambda = 5$ and threshold multiplier $c = 3$ (i.e., $m = 15$), 
we have $q \sim 10^{-4}$ (Section~\ref{sec:fundamental}). For populations of 
$n \sim 10^{6}$, the product $nq \sim 100 \gg 1$ places the system far above 
the critical scale.

In this regime, prior odds collapse faster than likelihood ratios can compensate. 
Even if the likelihood ratio $s/q$ is substantial, the prior odds $r/n$ are so 
unfavorable that posterior probabilities remain low. Stockmarr's caution applies: 
match evidence loses evidential weight as search populations grow.

\textbf{Resolution.}
The critical scale $n_{\mathrm{crit}} \sim \exp(\lambda D)$ separates these 
regimes. Stockmarr and Balding are both correct in their respective 
contexts: DNA forensics operates where $nq \ll 1$ (likelihood-driven), while 
multi-attribute surveillance operates where $nq \gg 1$ (base-rate dominated). 
The apparent contradiction dissolves once we recognize this regime distinction.

\subsubsection{Key Takeaways}

\begin{enumerate}
\item \textbf{Bayesian scaling mirrors frequentist scaling:} Bayesian actionability 
inherits the exponential factor $e^{\lambda D(c\|1)}$ but includes the additional 
multiplicative factor $(1-\alpha)rs/\alpha$; see \eqref{eq:bayes_critical_scale}.

\item \textbf{Posterior collapse can precede frequentist failure:} In the sparse-target 
regime, posterior reliability collapses once $nq$ approaches $rs$, often well before 
the frequentist transition at $nq \sim 1$.

\item \textbf{Exponential data growth overwhelms adaptation:} Reducing $q$ by a factor 
of $\beta$ requires increasing $k$ by only $\log(1/\beta)/(pD)$, while real-world data 
growth is exponential: $k(t)=k_0\gamma^{t}$ (Section~\ref{sec:temporal}). Thus 
posterior collapse is temporally inevitable.

\item \textbf{Epistemic saturation:} As $n$ grows, base rates $\pi=r/n$ shrink. Even 
rare false positives become dominated by prior odds, causing $\Pr(\text{target}\mid\text{flag})$ 
to decay toward zero.

\item \textbf{Resolution of the DNA debate:} Stockmarr and Balding are correct in 
different regimes: Balding for $nq \ll 1$ (DNA), Stockmarr for $nq \gtrsim 1$ (large-scale 
attribute screening).
\end{enumerate}

\begin{remark}[Connections to Classical Statistical Fallacies]
This section unifies the base-rate fallacy, the prosecutor's fallacy, false discovery 
rate control \cite{benjamini1995}, and the PPV problem in medical screening 
\cite{welch2011}. Conceptually, these phenomena are identical: all reflect Bayes' 
rule under low prevalence and imperfect specificity. The widely cited argument that 
``most published research findings are false'' \cite{ioannidis2005} is the same FDR/PPV 
problem in another domain.
\end{remark}

\subsection{Connections to Algorithmic Fairness}
\label{sec:fairness}

The Group Dominance Effect (Theorem~\ref{thm:dominance}) has important 
implications for fairness in surveillance systems. When different population 
groups experience differential surveillance exposure, small differences in 
exposure rates create exponential disparities in outcomes. 
Proposition~\ref{prop:exposure_amplification} shows that exposure ratios of 2--4 times generate false alert disparities exceeding 20 times near critical thresholds. This exponential amplification, which arises from Poisson tail behavior, means that even modest differences in surveillance intensity produce severe outcome inequalities.

Crucially, these disparities cannot be eliminated through threshold 
adjustment. Group-specific thresholds merely encode the underlying 
exposure inequality in a different form. Equalizing outcomes requires 
equalizing data collection intensity at the source, not algorithmic 
tuning. Moreover, since the high-exposure group drives system-level 
false alerts, aggregate reliability metrics obscure concentrated burdens 
on specific subpopulations, making demographic disaggregation essential 
for understanding actual system performance.

\begin{remark}[Structural vs.\ Algorithmic Bias]
\label{rem:structural_bias}
Proposition~\ref{prop:exposure_amplification} demonstrates that disparate 
outcomes arise from the \emph{probabilistic structure} of screening systems, 
independent of algorithmic design choices. When different groups experience 
differential surveillance exposure rates ($p_1 \neq p_2$), this mathematically 
guarantees unequal false positive rates ($q_1 \neq q_2$) under any common 
threshold $m$, creating disproportionate false alert burdens through Poisson 
tail behavior.

The exponential amplification in part (1) is particularly striking: when both 
groups are screened using the same attribute set and threshold, small 
differences in exposure translate to exponential differences in false alert 
rates (Figure~\ref{fig:phase_transitions}(d)). 

The effect manifests temporally as different groups reach critical false-alert 
rates at different times: Group~2 with exposure rate $p_2 = 0.020$ fails at 
$k_2^* = 350$ attributes, while Group~1 with $p_1 = 0.005$ remains reliable 
until $k_1^* = 1400$. While this fourfold difference in system lifetime simply 
reflects the fourfold difference in exposure rates (a linear relationship), the 
amplification becomes exponential when comparing simultaneous false alert rates: 
at any intermediate $k \in (350, 1400)$, Group~2 experiences exponentially more 
false alerts than Group~1.

\textbf{Connection to algorithmic fairness.} 
These findings relate directly to classical impossibility theorems in the 
algorithmic fairness literature. Kleinberg et al.\ \cite{kleinberg2017} and 
Chouldechova \cite{chouldechova2017} show that equalizing false positive 
rates, false negative rates, and calibration is impossible when base rates 
differ. Our analysis identifies a complementary, and more fundamental, 
mechanism: \emph{surveillance exposure itself creates different effective base 
rates} ($\lambda_g = k p_g$), guaranteeing unequal false alert burdens even 
\emph{before} any classifier is applied.

Standard fairness interventions operate at the classifier level: equalized odds 
\cite{hardt2016}, demographic parity, and calibration attempt to constrain 
predictions. None can correct the structural disparity we identify, because it 
arises from data collection intensity ($p_g$), not from how a classifier 
processes the collected data. Group-specific thresholds $m_g$ might equalize $q_1$ and $q_2$, but only by encoding exposure inequality directly. Achieving parity requires assigning a higher threshold to the more surveilled group.

This perspective connects to ``fairness through awareness'' \cite{dwork2012} and formalizes ``structural bias'' arguments from the critical algorithm studies literature \cite{eubanks2018,benjamin2019,mehrabi2021}. Disparities can be intrinsic to systems built on heterogeneous data collection, rather than artifacts of biased algorithms or training data.
\end{remark}

\begin{remark}[Policy Implications]
These results suggest that surveillance system audits should:
\begin{enumerate}
\item Measure exposure rates ($p_g$) across demographic and geographic groups, not just aggregate false alert rates.
\item Recognize that \emph{system reliability is bounded by the worst-performing group} (Theorem~\ref{thm:dominance}), making demographic disaggregation essential.
\item Account for temporal dynamics: groups with higher exposure fail first as data accumulates, creating windows of maximum disparity.
\item Acknowledge that threshold adjustments cannot eliminate disparities arising from differential exposure; only equalizing $p_g$ across groups can achieve fairness.
\end{enumerate}
\end{remark}

\subsection{Limitations and Future Directions}
\label{sec:limitations}

This analysis operates under several simplifying assumptions that define its scope and suggest natural directions for future work.

\textbf{Independence across individuals.} Our core results (Theorems~\ref{thm:critical_pop} and~\ref{thm:system_lifetime}) assume statistical independence of match counts across individuals (Remark~\ref{rem:independence}). Common-mode events (mass gatherings, natural disasters, viral social media content, coordinated activities) introduce positive dependence that would \emph{increase} false alert rates beyond our bounds. Positive dependence inflates upper-tail probabilities and therefore worsens system reliability relative to the independent case. Extensions via positively associated random variables or Chen--Stein methods \cite{barbour2005} could quantify these effects, but our independence-based analysis provides a lower bound on false alert rates.

\textbf{Binary attributes.} We model attributes as binary indicators (match/no-match). Continuous features, count data, or multi-level categorical attributes would require different distributional assumptions and tail bounds. The qualitative insights about combinatorial explosion in high-dimensional spaces should persist, but quantitative thresholds would differ. Extensions to Gaussian or sub-Gaussian attributes would preserve the essential exponential tail behavior underlying our critical-scale results.

\textbf{Fixed thresholds.} Our analysis assumes detection thresholds $m$ are fixed at deployment time. Adaptive systems that adjust thresholds based on observed alert rates or estimated base rates could potentially extend operational lifetimes. However, Theorem~\ref{thm:system_lifetime} suggests fundamental limits: under exponential data growth, even optimally adaptive thresholds would need to grow exponentially to maintain reliability, eventually exceeding meaningful detection capabilities.

\textbf{Constant sensitivity.} The Bayesian analysis (Section~\ref{sec:bayesian}) initially treats sensitivity $s = \Pr(\text{flag} \mid \text{target})$ as constant. Remark~\ref{rem:sensitivity_threshold} introduces a simple threshold-dependent model, but a complete analysis would require specifying both the target match distribution and the full signal distribution of true targets and performing ROC optimization \cite{fawcett2006}. This is inherently application-dependent and beyond our current scope.

\textbf{Heuristic correlation treatment.} The effective degrees of freedom 
approach (Section~\ref{sec:spatialtemporal}; Appendix~\ref{app:correlation}) 
captures variance inflation but does not constitute rigorous large-deviation 
analysis. Formal treatment would require specifying mixing conditions or 
dependency graph structures~\cite{dembo2010,barbour2005}. Our heuristic 
provides qualitative guidance rather than formal guarantees.

\textbf{Lack of empirical validation.} We have not validated our predictions against operational surveillance data, which is typically proprietary, classified, or subject to confidentiality restrictions. Instead, we use proxy datasets for illustrative rather than validating analyses (Appendix~\ref{app:empirical}); these provide qualitative checks but do not constitute validation in the intended deployment environment.

\textbf{Static population structure.} We assume fixed population composition with stable group sizes $n_g$ and exposure rates $p_g$. Dynamic populations with entry, exit, demographic shifts, and changing surveillance intensity would require stochastic population models. The temporal analysis (Section~\ref{sec:temporal}) addresses data growth but not population dynamics.

\textbf{Formal versus heuristic results.} Theorems~\ref{thm:critical_pop}, 
\ref{thm:system_lifetime}, and~\ref{thm:dominance}, along with 
Propositions~\ref{prop:sharp_threshold} and~\ref{prop:exposure_amplification}, 
are formal results with complete proofs. The effective dimensionality 
analysis (Appendix~\ref{app:correlation}) and sensitivity-threshold model 
(Remark~\ref{rem:sensitivity_threshold}) are heuristic approximations.

\textbf{Broader applicability.} While surveillance systems motivated this analysis, the mathematical framework applies to any domain where threshold rules screen large collections across many low-probability binary indicators. The critical population bounds (Theorem~\ref{thm:critical_pop}), temporal saturation dynamics (Theorem~\ref{thm:system_lifetime}), and group-level disparity amplification (Theorem~\ref{thm:dominance}) characterize generic properties of high-dimensional threshold detection, independent of the specific application. Natural extensions include network intrusion detection, manufacturing quality control, financial fraud screening, medical diagnostic panels, and environmental monitoring systems. The binary indicator assumption could be relaxed to accommodate hybrid frameworks combining discrete and continuous variables~\cite{wang2023}, though the essential combinatorial explosion in threshold-based screening would persist. We developed the theory through the surveillance lens because it offered the clearest exposition of the societal stakes, but the probabilistic limits derived here constrain any system that aggregates rare coincidences across high-dimensional attribute spaces.

Despite these limitations, the core mathematical structure (exponential scaling of critical populations, finite system lifetimes under data growth, and structural amplification of exposure disparities) should prove robust across modeling variations. Modeling refinements would shift numerical thresholds but not the qualitative scaling laws, which arise from intrinsic high-dimensional coincidence phenomena.

\section{Conclusions}

We have established sharp probabilistic limits, under standard rare-event and independence (or effective-independence) assumptions, for large-scale screening systems. The critical population size beyond which false alerts become inevitable scales as
\[
n_{\mathrm{crit}} \sim \sqrt{\lambda}\exp(\lambda D(c\|1)),
\]
governed by large-deviation rate functions that cannot be circumvented through algorithmic refinement within these models. When data volume grows exponentially, $k(t) = k_0\gamma^{t}$, any screening system has a calculable operational lifetime
\[
T^{*} \approx \frac{\log m - \log(k_{0}p)}{\log \gamma},
\]
typically measured in years rather than decades. Threshold adjustments provide only logarithmic protection against exponential growth, making temporal failure inevitable rather than merely possible.

The Bayesian analysis reveals an even more stringent constraint: posterior probabilities collapse when $nq \gg rs$, rendering flags epistemically meaningless well before frequentist reliability fails. This asymptotic framing clarifies the long-standing DNA database controversy: Stockmarr's caution and Balding's confidence apply in distinct regimes, separated by the critical scale $n \sim \exp(\lambda D)$.

Differential surveillance exposure creates exponential outcome disparities whenever common thresholds are applied across groups with heterogeneous match probabilities $p_g$. The mathematics guarantees disproportionate false-alert burdens through Poisson tail behavior. This is not an artifact of algorithmic bias but a structural inevitability arising from heterogeneous data collection.

These results impose non-negotiable constraints on system design. Operational surveillance systems require: (1) explicit expiration dates calculated from data growth rates; (2) demographic exposure audits measuring $p_g$ across population groups; (3) capacity constraints limiting investigable flags; and (4) recognition that system-wide reliability is dominated by the subpopulation with the highest effective match probability. The mathematics is unforgiving: more data does not guarantee better decisions, and exponential growth ensures finite operational windows.

The policy implication is direct: screening systems operating near or beyond their critical parameters will generate false alerts regardless of implementation quality. Designers must either accept these mathematical limits or fundamentally restructure surveillance architectures to avoid accumulating correlated data over time. Threshold refinement and algorithmic optimization cannot solve problems rooted in probabilistic inevitability.

\section*{Acknowledgments}

We acknowledge the support of the Natural Sciences and Engineering Research Council of Canada (NSERC), funding reference number RGPIN-2019-04085.

\section*{Data Availability}

The illustrative examples in Appendix~\ref{app:empirical} use two 
publicly available datasets: the UCI Adult Census dataset 
(\url{https://archive.ics.uci.edu/dataset/2/adult}) and the City of Chicago 
Crime dataset (\url{https://data.cityofchicago.org/Public-Safety/Crimes-One-year-prior-to-present/x2n5-8w5q}). 
No new data were generated.

\section{Notation Summary}
\label{sec:notation}

\begin{tabular}{@{}ll@{}}
\toprule
\textbf{Symbol} & \textbf{Description} \\
\midrule
$n$ & Population size (number of individuals screened) \\
$k$ & Total number of binary attributes monitored (before dependence) \\
$p$ & Marginal per-attribute match probability for an innocent individual \\
$\lambda$ & Expected false matches under independence; $\lambda = kp$ \\
$m$ & Detection threshold (individual flagged if $X_i \ge m$) \\
$c$ & Threshold multiplier; typically $m \approx c\lambda$ for $c>1$ \\
$q$ & Per-person false alert probability; $q = \Pr(X \ge m \mid \text{innocent})$ \\
$D(\alpha \| 1)$ & Poisson rate function; $D(\alpha\|1) = \alpha\log\alpha - \alpha + 1$ \\
$n_{\mathrm{crit}}$ & Population size where expected false alerts reach order one \\
$s$ & Sensitivity (true positive rate); $s = \Pr(\text{flag} \mid \text{target})$ \\
$r$ & Expected number of true targets in the population \\
$\pi$ & Base rate (prevalence); $\pi = r/n$ \\
$p_g$ & Marginal match probability for group $g$ \\
$n_g$ & Size of group $g$ \\
$\lambda_g$ & Expected false matches for group $g$; $\lambda_g = k p_g$ \\
$q_g$ & Per-person false alert probability for group $g$ \\
$k_{\mathrm{eff}}$ & Effective number of independent attributes accounting for correlation \\
$\gamma$ & Data growth factor per unit time \\
$T^*$ & System lifetime; time at which $\lambda(t)$ reaches the threshold $m$ \\
$X_i$ & Number of matching attributes for individual $i$ \\
$\alpha$ & Posterior probability threshold in Bayesian analysis \\
$\xi$ & Spatial correlation length \\
$\tau$ & Temporal correlation time \\
\bottomrule
\end{tabular}

\appendix

\section{Detailed Proofs}
\label{app:proofs}

This appendix provides complete proofs with all intermediate steps for the main theoretical results.

\subsection{Proof of Theorem~\ref{thm:critical_pop} (Critical Population Scale)}
\label{app:proof_critical_pop}

Recall that $Y \sim \mathrm{Poisson}(\lambda)$, $\lambda = kp$, and 
$m = \lceil c\lambda \rceil$ with $c>1$. We write 
$q = \Pr(Y \ge m)$ for the Poisson upper tail. Replacing $m$ by $c\lambda$ 
(ignoring integer rounding) affects only multiplicative constants in the bounds 
and does not change the exponential rate.

\medskip
\textbf{Upper bound on $q$.}
Lemma~\ref{lem:poisson_upper} (Chernoff bound for Poisson tails) applied with 
$m = c\lambda$ yields
\[
q = \Pr(Y \ge m) \le \exp(-\lambda D(c\|1)),
\]
where the rate function is
\[
D(c\|1) = c\log c - c + 1.
\]

\medskip
\textbf{Lower bound on $q$.}
We obtain a lower bound on $\Pr(Y = m)$ using Robbins' refinement of Stirling's
formula \cite{robbins1955}:
\[
m! < \sqrt{2\pi m}\left(\frac{m}{e}\right)^m e^{1/(12m)}.
\]
Hence
\begin{align*}
\Pr(Y = m)
&= e^{-\lambda}\frac{\lambda^m}{m!} \\
&\ge \frac{e^{-\lambda}\lambda^m}{\sqrt{2\pi m}\left(\frac{m}{e}\right)^m e^{1/(12m)}} \\
&= \frac{1}{\sqrt{2\pi m}}
   \exp\!\Big(-\lambda + m\log\lambda - m\log m + m - \tfrac{1}{12m}\Big).
\end{align*}
Now set $m = c\lambda$ and simplify the exponent:
\begin{align*}
-\lambda + m\log\lambda - m\log m + m
&= -\lambda + c\lambda\log\lambda - c\lambda\log(c\lambda) + c\lambda \\
&= -\lambda + c\lambda\log\lambda - c\lambda\log c - c\lambda\log\lambda + c\lambda \\
&= -\lambda - c\lambda\log c + c\lambda \\
&= -\lambda\big(1 + c\log c - c\big) \\
&= -\lambda\big(c\log c - c + 1\big) \\
&= -\lambda D(c\|1).
\end{align*}
Substituting this back, and using $m = c\lambda$, we obtain
\[
\Pr(Y=m)
\ge
\frac{1}{\sqrt{2\pi c\lambda}}\,
\exp\!\Big(-\lambda D(c\|1) - \tfrac{1}{12c\lambda}\Big).
\]
Since $q = \Pr(Y \ge m) \ge \Pr(Y = m)$, the same lower bound applies to $q$,
establishing the two-sided tail estimate in \eqref{eq:poisson-tail-two-sided}:
\[
\frac{1}{\sqrt{2\pi c\lambda}}\,
\exp\!\Big(-\lambda D(c\|1) - \tfrac{1}{12c\lambda}\Big)
\;\le\; q \;\le\; \exp(-\lambda D(c\|1)).
\]

\medskip
\textbf{System-level bounds.}
Independence across $n$ screened individuals implies
\[
\Pr(\text{no false alert}) = (1-q)^n,
\qquad
\Pr(\text{false alert}) = 1 - (1-q)^n.
\]
To relate $(1-q)^n$ to exponentials, we use standard inequalities for 
$x \in (0,1)$:
\[
\ln(1-x) \le -x
\quad\text{and}\quad
\ln(1-x) \ge -\frac{x}{1-x}.
\]
The first inequality follows from the concavity of $\ln$ and the tangent bound
$\ln(1-x) \le -x$ at $x=0$; the second can be verified by comparing power-series
expansions or by bounding the remainder term in the Taylor series. Exponentiating
and raising to the $n$th power gives
\[
e^{-nx/(1-x)} \le (1-x)^n \le e^{-nx},
\]
and hence, with $x = q$,
\[
e^{-nq/(1-q)} \le (1-q)^n \le e^{-nq}.
\]
Therefore,
\[
1 - e^{-nq} \;\le\; \Pr(\text{false alert}) \;\le\;
1 - e^{-nq/(1-q)},
\]
which are precisely the system-level bounds \eqref{eq:system-bounds} once the 
Poisson tail bounds for $q$ are substituted.

\medskip
\textbf{Explicit form of system bounds.}
Substituting the Poisson tail bounds into the system-level inequalities yields
the explicit form \eqref{eq:system-bounds}. For the lower bound on 
$\Pr(\text{false alert})$, we use $1 - e^{-nq}$ together with the lower bound on $q$:
\[
\Pr(\text{false alert}) \ge 1 - \exp\!\left(
  -\frac{n}{\sqrt{2\pi c\lambda}}\,
   e^{-\lambda D(c\|1) - 1/(12c\lambda)}
\right).
\]
For the upper bound, we use $1 - e^{-nq/(1-q)}$ with the upper bound 
$q \le e^{-\lambda D(c\|1)}$. Since $q/(1-q)$ is increasing in $q$ for $q \in (0,1)$,
\[
\frac{q}{1-q} \le \frac{e^{-\lambda D(c\|1)}}{1 - e^{-\lambda D(c\|1)}},
\]
giving
\[
\Pr(\text{false alert}) \le 1 - \exp\!\left(
  -\frac{n\,e^{-\lambda D(c\|1)}}{1 - e^{-\lambda D(c\|1)}}
\right).
\]

\medskip
\textbf{Critical population scale.}
The system becomes unreliable once we expect on the order of one 
false alert per batch. Using the lower bound on $\Pr(Y=m)$ as a proxy for $q$,
this corresponds to
\[
n\,\Pr(Y=m)
\approx
\frac{n}{\sqrt{2\pi c\lambda}}\,
\exp\!\Big(-\lambda D(c\|1) - \tfrac{1}{12c\lambda}\Big)
\sim 1.
\]
Solving for $n$ gives
\[
n_{\mathrm{crit}} \approx \sqrt{2\pi c\lambda}\,
\exp\!\Big(\lambda D(c\|1) + \tfrac{1}{12c\lambda}\Big).
\]
For large $\lambda$, the correction term $1/(12c\lambda)$ is negligible, and we
obtain the asymptotic form
\[
n_{\mathrm{crit}} \asymp \sqrt{\lambda}\, e^{\lambda D(c\|1)}.
\]
Since the factor $\sqrt{\lambda}$ is subexponential in $\lambda$, the dominant 
scaling is
\[
n_{\mathrm{crit}} \asymp e^{\lambda D(c\|1)} = e^{kp\,D(c\|1)},
\]
which is the critical population size stated in 
Theorem~\ref{thm:critical_pop}.

\subsection{Proof of Proposition~\ref{prop:sharp_threshold} (Threshold Sharpness)}
\label{app:proof_sharp_threshold}

From Theorem~\ref{thm:critical_pop}, we have the two-sided bounds
\[
\frac{1}{\sqrt{2\pi c\lambda}}\,e^{-\lambda D(c\|1) - O(1/\lambda)} 
\le q \le e^{-\lambda D(c\|1)},
\]
which together imply $\log q = -\lambda D(c\|1) + O(\log\lambda)$ as 
$\lambda \to \infty$.

With $n=\sqrt{\lambda}\,e^{\alpha\lambda D(c\|1)}$, we compute
\begin{align*}
\log(nq) &= \log n + \log q \\
&= \tfrac{1}{2}\log\lambda + \alpha\lambda D(c\|1) - \lambda D(c\|1) + O(\log\lambda) \\
&= (\alpha-1)\lambda D(c\|1) + O(\log\lambda).
\end{align*}

For $\alpha < 1$: the leading term $(\alpha-1)\lambda D < 0$ dominates, so 
$nq \to 0$ exponentially fast.

For $\alpha > 1$: the leading term $(\alpha-1)\lambda D > 0$ dominates, so 
$nq \to \infty$ exponentially fast.

The transition width $\Delta\alpha = O(1/(\lambda D))$ vanishes as 
$\lambda \to \infty$. The limiting behavior~\eqref{eq:sharp_threshold} 
follows from $\Pr(\text{false alert}) \approx 1 - e^{-nq}$.

\subsection{Proof of Proposition~\ref{prop:exposure_amplification} 
(Exposure Amplification)}
\label{app:proof_exposure}

\textbf{Part (1):} When $\lambda_2 \ge m$, monotonicity gives
\[
q_2 = \Pr(\text{Poisson}(\lambda_2) \ge m) \ge \Pr(\text{Poisson}(m) \ge m).
\]
For large $m$, normal approximation gives 
$\Pr(\text{Poisson}(m) \ge m) \ge 1/2 - O(m^{-1/2})$.

Meanwhile, $q_1 \le \exp(-\lambda_1 D(m/\lambda_1\|1))$ by 
Lemma~\ref{lem:poisson_upper}, which is exponentially small when 
$m > \lambda_1$. The ratio $q_2/q_1$ therefore grows at least exponentially: 
for any $c > 0$, there exists $M(c)$ such that $q_2/q_1 \ge e^{cm}$ for 
$m \ge M(c)$.

\textbf{Part (2):} Since $q_2/q_1 \to \infty$, we have 
$q_2/(q_1 + q_2) \to 1$.

\textbf{Part (3):} Group~2 reaches $\lambda_2 = m$ at $k = m/p_2 < m/p_1$. 
The disparity is maximal for $k \in (m/p_2, m/p_1)$.

\section{Effective Dimensionality Under Correlation: Technical Details}
\label{app:correlation}

This appendix provides the full technical development of correlation effects 
summarized in Section~\ref{sec:spatialtemporal}.

\subsection{Variance Inflation and Rate Function Reduction}

For binary indicators $X_1,\ldots,X_k$ with $\mathbb{E}[X_i] = p$ and 
pairwise correlation $\mathrm{Corr}(X_i, X_j) = \rho_{ij}$, the sum 
$Y = \sum_{i=1}^k X_i$ has:
\begin{align*}
\mathbb{E}[Y] &= kp \quad \text{(unchanged by correlation)}, \\
\mathrm{Var}(Y) &= kp(1-p) + \sum_{i \neq j} p(1-p) \rho_{ij} 
= kp(1-p)\left(1 + \frac{\sum_{i \neq j} \rho_{ij}}{k}\right),
\end{align*}
where the \emph{design effect} 
$\mathrm{DEFF} = \frac{\sum_{i \neq j} \rho_{ij}}{k}$ quantifies variance 
inflation; for nonnegative average correlation this satisfies 
$\mathrm{DEFF} \ge 0$~\cite{kish1965}.

For the Poisson approximation, independence gives 
$\mathrm{Var}(Y) = \lambda = kp$. Under positive correlation, 
$\mathrm{Var}(Y) > \lambda$, creating heavier tails.

Positive correlation generally reduces the large-deviation rate compared to 
the independent case; see~\cite{dembo2010} for precise conditions. 
Intuitively, the rate reduction can be modeled by shrinking the exponent 
by the factor $k_{\mathrm{eff}}/k$.

\subsection{Effective Degrees of Freedom}

Rather than reducing $\lambda$ itself, correlation reduces the 
\emph{effective degrees of freedom} governing concentration behavior. 
Standard design-effect methodology~\cite{kish1965,bretherton1999} 
parameterizes this through an effective sample size $k_{\mathrm{eff}}$ 
defined such that the variance of the correlated sum equals that of 
$k_{\mathrm{eff}}$ independent observations:
\begin{equation}
\mathrm{Var}(Y) = k_{\mathrm{eff}} \cdot p(1-p).
\end{equation}

For spatial correlation with exponential decay $\rho(r) = e^{-r/\xi}$ over 
a monitored area $A$, the spatial integral of the correlation function 
yields:
\begin{equation}
k_{\mathrm{eff}} \approx \frac{A}{2\pi\xi^2},
\end{equation}
where $\xi$ is the correlation length~\cite{bretherton1999}.

Similarly, for temporal correlation with correlation time $\tau$ and unit 
sampling interval, the Bartlett--Wilks formula~\cite{wilks2019} gives:
\begin{equation}
k_{\mathrm{eff}} \approx 
\frac{k}{1 + 2\sum_{h=1}^{k-1}\rho_{\mathrm{time}}(h)(1-h/k)}.
\end{equation}
For exponential temporal correlation $\rho_{\mathrm{time}}(h) = e^{-h/\tau}$ 
with $\tau \gg 1$:
\begin{equation}
k_{\mathrm{eff}} \approx \frac{k(1-e^{-1/\tau})}{1+e^{-1/\tau}} 
\approx \frac{k}{2\tau}.
\end{equation}

\subsection{Impact on Critical Populations}

We model the effect of correlation by shrinking the effective number of 
independent comparisons from $k$ to $k_{\mathrm{eff}}$, holding the marginal 
match probability $p$ fixed. The true mean number of matches remains 
$\lambda = kp$, but the concentration (and hence the large-deviation 
exponent) behaves as if only $k_{\mathrm{eff}}$ coordinates contributed 
independently.

For thresholds $m = c\lambda$ with $c > 1$, the tail probability under 
correlation satisfies:
\begin{equation}
q_{\mathrm{corr}} \gtrsim \exp\!\left(-k_{\mathrm{eff}} p \cdot D(c\|1)\right) 
= \exp\!\left(-\lambda \cdot \frac{k_{\mathrm{eff}}}{k} \cdot D(c\|1)\right).
\end{equation}

This modifies the critical population scale. The exponent in the independent 
case $\lambda D(c\|1) = kpD(c\|1)$ is reduced by the factor 
$k_{\mathrm{eff}}/k$, yielding:
\begin{equation}
n_{\mathrm{crit}}^{\mathrm{corr}} \lesssim \sqrt{\lambda} 
\exp\!\left(\frac{k_{\mathrm{eff}}}{k} \cdot \lambda D(c\|1)\right) 
= \sqrt{\lambda} \exp\!\left(k_{\mathrm{eff}} p \cdot D(c\|1)\right).
\end{equation}

The critical population is thus \emph{reduced} compared to the independent 
case, making systems fail at smaller populations when positive correlations 
are present. The factor $k_{\mathrm{eff}}/k$ represents the 
\emph{information loss} due to redundancy in correlated observations.

\subsection{Quantitative Examples}

\begin{example}[City-Scale Spatial Surveillance]
Consider fine-grained location monitoring with $k = 10{,}000$ cells covering 
area $A = 100$~km$^2$ with correlation length $\xi = 500$~m. The effective 
number of independent spatial tests is:
\[
k_{\mathrm{eff}} \approx \frac{100 \times 10^6 \text{ m}^2}{2\pi (500)^2} 
\approx 64.
\]
Despite nominally tracking 10,000 locations, correlation reduces effective 
information to approximately 64 independent observations.

With $p = 0.005$, the expected number of matches is 
$\lambda = kp = 10{,}000 \times 0.005 = 50$. Set the threshold at $m = 75$ 
matches (corresponding to $c = m/\lambda = 1.5$). The rate function is:
\[
D(1.5\|1) = 1.5\log(1.5) - 1.5 + 1 \approx 0.108.
\]

\textbf{Independent case:} The exponent governing the critical population is:
\[
\lambda D(c\|1) = 50 \times 0.108 = 5.4,
\]
giving $n_{\mathrm{crit}} \sim \sqrt{\lambda}\exp(5.4) \approx 1{,}560$.

\textbf{Correlated case:} The exponent is reduced by the factor 
$k_{\mathrm{eff}}/k$:
\[
\frac{k_{\mathrm{eff}}}{k} \cdot \lambda D(c\|1) 
= \frac{64}{10{,}000} \times 5.4 \approx 0.035,
\]
giving $n_{\mathrm{crit}}^{\mathrm{corr}} \sim \sqrt{\lambda}\exp(0.035) 
\approx 7$.

Spatial correlation reduces the critical population from approximately 
1,560 to approximately 7, a reduction of over two orders of magnitude.
\end{example}

\begin{example}[Routine Behavior Monitoring]
For daily surveillance data ($\Delta t = 1$ day) over one year 
($k = 365$ observations) monitoring movement patterns with correlation 
time $\tau = 30$ days:
\[
k_{\mathrm{eff}} \approx \frac{k}{2\tau} = \frac{365}{60} \approx 6.
\]
Observing someone's location 365 times yields effective information 
equivalent to approximately 6 independent observations.

With per-observation match probability $p = 0.02$, the expected matches 
are $\lambda = 365 \times 0.02 = 7.3$. Set threshold $m = 12$ 
(giving $c \approx 1.64$). The rate function is 
$D(1.64\|1) \approx 0.173$.

Independent case: $\lambda D(c\|1) \approx 1.26$, giving 
$n_{\mathrm{crit}} \approx 10$.

Correlated case: $(k_{\mathrm{eff}}/k) \cdot \lambda D \approx 0.021$, 
giving $n_{\mathrm{crit}}^{\mathrm{corr}} \approx 3$.

Temporal correlation reduces the critical population from approximately 
10 individuals to approximately 3, making reliable surveillance of even 
small groups challenging.
\end{example}

\subsection{Connections to Multiple Testing}

This effective dimension reduction connects surveillance analysis to 
multiple-testing corrections in genomics~\cite{nyholt2004,liji2005,uffelmann2021} and 
neuroimaging~\cite{worsley1996,friston1994,nichols2003}, where similar correlation structures 
require Bonferroni-type corrections scaled by $k_{\mathrm{eff}}$ rather 
than the nominal number of tests $k$. The mathematics forces system 
designers to confront the limited information content of nominally 
``big'' data: more observations do not guarantee more information when 
those observations are highly correlated.

\begin{remark}[Overdispersion as System Weakness]
Correlation-induced overdispersion represents a fundamental vulnerability 
in screening systems. Positive correlation makes extreme events (high match 
counts) more probable than the independent Poisson model predicts, 
accelerating the onset of false-alert saturation. Critically, system 
designers cannot escape this by collecting more data: additional correlated 
observations provide diminishing marginal information while accumulating the 
same false-positive burden.
\end{remark}

Appendix~\ref{app:empirical} illustrates these correlation effects using the
City of Chicago Crime dataset, where strong spatial clustering reduces the
effective dimensionality by roughly two orders of magnitude.

\section{Empirical Illustrations}
\label{app:empirical}

The theoretical results in this paper are derived independently of any data.
The empirical examples below illustrate how the predicted behaviors appear in
realistic settings and help clarify the practical implications of the theory.
These illustrations do not affect the validity of the theoretical results but
provide concrete context for their interpretation.

Two public datasets are used throughout: one providing individual-level
attribute profiles, the other providing population-level event patterns across
space and time. Together, they illustrate the two components common to screening
problems: a collection of individuals described by multiple attributes, and a
collection of events observed across space and time. The mathematical structure
examined in this paper captures the behavior of comparisons between these
domains in the rare-event regime.

The datasets are analyzed independently; no cross-dataset linkage is performed.
The UCI Adult Census dataset \cite{dua2019} illustrates individual-level
phenomena (match distributions, group disparities, false alert concentration),
while the City of Chicago Crime dataset \cite{chicagocrime2024} illustrates
system-level dynamics (temporal saturation, spatial clustering, correlation
effects).

\subsection{A Screening System from Census Data}

To illustrate individual-level screening behavior, we constructed a system from
the UCI Adult dataset (30,162 individuals after removing incomplete records).
The system monitors $k = 15$ binary indicators representing potential ``flags''
across financial, occupational, educational, and demographic dimensions
(Table~\ref{tab:screening_attributes}). The 15 indicators were selected to span
a range of prevalences and correlations, mirroring the heterogeneous binary
features often used in real screening systems.

An individual is flagged when the number of matching attributes equals or
exceeds a threshold $m$. Because the UCI dataset contains no true targets,
every individual is treated as innocent, so all flags are interpreted as false
alerts in the sense of the theoretical model.

\begin{table}[htbp]
\caption{Screening attributes constructed from UCI Adult Census data.}
\label{tab:screening_attributes}
\centering
\begin{tabular}{llc}
\toprule
Category & Attribute & Prevalence \\
\midrule
Financial   & High capital gains ($>$\$5000)     & 5.2\% \\
            & Any capital loss                    & 4.7\% \\
            & High income ($>$\$50K)              & 24.9\% \\
\midrule
Work        & Self-employed                       & 11.8\% \\
            & Overtime hours ($>$50/week)         & 11.5\% \\
            & Part-time ($<$35 hours/week)        & 15.5\% \\
\midrule
Occupation  & Executive/Managerial                & 13.2\% \\
            & Professional specialty              & 13.4\% \\
            & Sales                               & 11.9\% \\
            & Tech support                        & 3.0\% \\
\midrule
Education   & Graduate degree                     & 8.4\% \\
            & Bachelor's degree                   & 16.7\% \\
\midrule
Demographic & Age 35--55                          & 45.8\% \\
            & Married                             & 47.9\% \\
            & Foreign-born                        & 8.8\% \\
\bottomrule
\end{tabular}
\end{table}

For each individual, we computed the number of matching attributes. The
resulting distribution has mean $\lambda = 2.43$ and variance $\sigma^2 = 2.95$,
yielding a variance-to-mean ratio of 1.22. Comparing to the Poisson$(\lambda)$
distribution gives total variation distance $d_{\mathrm{TV}} = 0.084$. This is
larger than Monte Carlo simulations of independent attributes
($d_{\mathrm{TV}} \approx 0.005$) but still indicates reasonable model fit; the
deviation reflects the correlation structure present in real demographic
data, the phenomenon analyzed in Section~\ref{sec:spatialtemporal}.

Table~\ref{tab:flagging_rates} reports the fraction of the population flagged at
each threshold.

\begin{table}[htbp]
\caption{False positive rates from UCI Adult data at each threshold.}
\label{tab:flagging_rates}
\centering
\begin{tabular}{ccc}
\toprule
Threshold $m$ & \% Population Flagged & Per Million \\
\midrule
3 & 41.7\% & 417,000 \\
4 & 25.2\% & 252,000 \\
5 & 13.8\% & 138,000 \\
6 & 5.9\%  & 59,000 \\
7 & 1.7\%  & 17,000 \\
8 & 0.31\% & 3,100 \\
\bottomrule
\end{tabular}
\end{table}

\subsection{Poisson Model Fit}

Table~\ref{tab:poisson_fit} compares empirical false alert rates from the UCI
Adult data against predictions from a Poisson$(\lambda)$ model with
$\lambda = 2.43$. The ratio of empirical to theoretical probabilities quantifies
model fit across thresholds.

\begin{table}[htbp]
\caption{Empirical versus Poisson-predicted false alert rates from UCI Adult
census data ($\lambda = 2.43$, $n = 30{,}162$).}
\label{tab:poisson_fit}
\centering
\begin{tabular}{cccc}
\toprule
Threshold $m$ & Empirical $q$ & Poisson $q$ & Ratio \\
\midrule
3 & 41.7\% & 43.8\% & 0.95 \\
4 & 25.2\% & 22.7\% & 1.11 \\
5 & 13.8\% & 10.0\% & 1.38 \\
6 & 5.9\%  & 3.7\%  & 1.56 \\
7 & 1.7\%  & 1.2\%  & 1.39 \\
8 & 0.3\%  & 0.4\%  & 0.87 \\
\bottomrule
\end{tabular}
\end{table}

At intermediate thresholds ($m = 4$--$7$), empirical rates exceed Poisson
predictions by a factor of 1.1--1.6, reflecting the overdispersion induced by
attribute correlation (variance-to-mean ratio $= 1.22$ versus $1.0$ for
Poisson). This is precisely the phenomenon analyzed in
Section~\ref{sec:spatialtemporal}: positive correlation inflates tail
probabilities, causing false alerts to occur more frequently than independence
would predict. The effect is most pronounced at moderate thresholds where the
tail is neither too heavy nor too light.

At extreme thresholds ($m = 3$ and $m = 8$), the ratio approaches unity,
indicating that the Poisson approximation remains reasonable for order-of-magnitude
estimates even when correlation is present. The total variation distance
$d_{\mathrm{TV}} = 0.084$ reported earlier reflects this moderate deviation.

\subsection{Group Disparities}

The UCI dataset permits analysis of screening outcomes across demographic
groups. Table~\ref{tab:group_exposure} reports the mean number of attribute
matches ($\lambda_g$) for race-by-sex groups with sample size $n_g \geq 100$.
Exposure varies by a factor of 2.5 across groups.

\begin{table}[htbp]
\caption{Group-level exposure in census data.}
\label{tab:group_exposure}
\centering
\begin{tabular}{lcc}
\toprule
Group & $n_g$ & $\lambda_g$ \\
\midrule
Asian-Pacific Islander, Male   & 601      & 3.68 \\
White, Male                    & 18,038 & 2.74 \\
Asian-Pacific Islander, Female & 294      & 2.69 \\
Black, Male                    & 1,418  & 1.91 \\
White, Female                  & 7,895  & 1.90 \\
Black, Female                  & 1,399  & 1.49 \\
\bottomrule
\end{tabular}
\end{table}

As Proposition~\ref{prop:exposure_amplification} predicts, the 2.5-fold exposure
difference translates to exponentially larger disparities in false alert rates
as the threshold increases (Table~\ref{tab:disparity_amplification}). Linear
regression of $\log(q_{\mathrm{high}}/q_{\mathrm{low}})$ on $m$ yields slope
$= 0.78$, $R^2 = 0.976$, $p = 0.0016$, confirming the exponential amplification
mechanism. The empirical slope is consistent with the theoretical rate function
$D(c\|1)$ governing exponential tail decay.

\begin{table}[htbp]
\caption{Disparity amplification: comparing highest-exposure group
(Asian-Pacific Islander Male, $\lambda = 3.68$) to lowest (Black Female,
$\lambda = 1.49$).}
\label{tab:disparity_amplification}
\centering
\begin{tabular}{ccccc}
\toprule
$m$ & $q_{\mathrm{high}}$ & $q_{\mathrm{low}}$ & Ratio & $\log(\text{Ratio})$ \\
\midrule
3 & 0.696 & 0.169 & 4.1   & 1.41 \\
4 & 0.513 & 0.059 & 8.7   & 2.17 \\
5 & 0.343 & 0.022 & 15.5  & 2.74 \\
6 & 0.195 & 0.006 & 30.3  & 3.41 \\
7 & 0.077 & 0.001 & 107.1 & 4.67 \\
\bottomrule
\end{tabular}
\end{table}

\subsection{Group Dominance}

Theorem~\ref{thm:dominance} predicts that system-level false alerts concentrate
in whichever group maximizes $n_g q_g$, and that this concentration intensifies
as thresholds increase. Table~\ref{tab:group_dominance} illustrates this pattern
clearly. At $m = 5$, alert shares roughly track population shares. By $m = 8$,
small high-exposure groups are dramatically over-represented while
low-exposure groups virtually disappear.

\begin{table}[htbp]
\caption{Group dominance intensifies with threshold. Asian-Pacific Islander
males (2\% of population, highest $\lambda_g$) account for 17\% of false alerts
at $m = 8$; Black females (4.6\% of population, lowest $\lambda_g$) account for
none.}
\label{tab:group_dominance}
\centering
\begin{tabular}{lccccc}
\toprule
Group & Pop.\ \% & $\lambda_g$ & \multicolumn{3}{c}{Alert share (\%)} \\
\cmidrule(lr){4-6}
      &          &             & $m=5$ & $m=7$ & $m=8$ \\
\midrule
Asian-Pac-Isl., Male   & 2.0  & 3.68 & 5.0  & 8.9  & \textbf{17.0} \\
White, Male            & 59.8 & 2.74 & 78.9 & 78.9 & 72.3 \\
White, Female          & 26.2 & 1.90 & 11.1 & 8.5  & 5.3  \\
Black, Male            & 4.7  & 1.91 & 2.5  & 1.2  & 1.1  \\
Black, Female          & 4.6  & 1.49 & 0.7  & 0.2  & \textbf{0.0}  \\
Other groups           & 2.7  & ---  & 1.8  & 2.3  & 4.3  \\
\bottomrule
\end{tabular}
\end{table}

The highest-exposure group (Asian-Pacific Islander males, $\lambda_g = 3.68$)
represents only 2\% of the population but accounts for 17\% of false alerts at
$m = 8$, an 8.5-fold amplification. Conversely, Black females
($\lambda_g = 1.49$) represent 4.6\% of the population yet generate effectively
zero false alerts at the same threshold.

This divergence reflects the exponential sensitivity of $q_g$:
groups with slightly larger $\lambda_g$ retain non-negligible tail probability
as $m$ increases, while groups with smaller $\lambda_g$ collapse exponentially.
At high thresholds, contributions to $n_g q_g$ no longer resemble population
shares; exposure, not size, determines system-wide false alerts, precisely as
predicted by Theorem~\ref{thm:dominance} and
Proposition~\ref{prop:exposure_amplification}.

\subsection{Temporal Saturation from Crime Data}

The Chicago Crime dataset (236,967 incidents across 50 wards over one year)
permits direct observation of temporal saturation. Table~\ref{tab:temporal_saturation}
tracks the fraction of wards exceeding fixed crime-count thresholds as weeks
accumulate. Thresholds were chosen to span the range of cumulative ward-level
crime counts (approximately 2,000--8,000 per ward annually).

Here the same symbol $m$ is used for the alert threshold, but the underlying
quantity is cumulative ward-level crime counts rather than attribute matches.
Accordingly, the numerical scale of $m$ is much larger, though its mathematical
role in determining flags is identical.

\begin{table}[htbp]
\caption{Temporal saturation: fraction of Chicago wards exceeding threshold.}
\label{tab:temporal_saturation}
\centering
\begin{tabular}{ccccc}
\toprule
Week & $m = 1000$ & $m = 2000$ & $m = 3000$ & $m = 4000$ \\
\midrule
4   & 0\%    & 0\%   & 0\%   & 0\%  \\
8   & 22\%   & 0\%   & 0\%   & 0\%  \\
12  & 38\%   & 4\%   & 0\%   & 0\%  \\
16  & 58\%   & 28\%  & 4\%   & 0\%  \\
24  & 98\%   & 38\%  & 28\%  & 6\%  \\
32  & 100\%  & 62\%  & 38\%  & 30\% \\
40  & 100\%  & 80\%  & 44\%  & 34\% \\
48  & 100\%  & 98\%  & 62\%  & 38\% \\
\bottomrule
\end{tabular}
\end{table}

Every threshold eventually approaches saturation: raising thresholds delays but
does not prevent the transition, consistent with Theorem~\ref{thm:system_lifetime}.
The time to 50\% saturation ($T_{50}$) scales approximately linearly with
threshold: $T_{50} \approx 7$ weeks for $m = 1000$, $\approx 14$ weeks for
$m = 1500$, $\approx 29$ weeks for $m = 2000$. Saturation here refers to
cumulative counts exceeding fixed thresholds, not an increase in crime rates;
even a stationary process saturates any fixed threshold given sufficient time.

\subsection{Correlation Structure}

Crime counts across Chicago's 50 wards exhibit strong spatial overdispersion:
mean $= 4{,}739$ crimes per ward, variance $= 5{,}996{,}000$, yielding
variance-to-mean ratio $= 1{,}265$. For Poisson-distributed data this ratio
equals 1; the observed value indicates substantial positive spatial correlation
(crime clustering), reducing effective dimensionality as analyzed in
Section~\ref{sec:spatialtemporal}.

Daily crime counts exhibit temporal autocorrelation with estimated correlation
time $\tau \approx 15$ days. By the Bartlett--Wilks formula, this implies
$k_{\mathrm{eff}} \approx 365/(2\tau) \approx 12$ effective independent
observations per year, a reduction factor of approximately 30. This reduction
in effective dimensionality directly illustrates the correlation correction
derived in Section~\ref{sec:spatialtemporal}.

\subsection{Summary}

Table~\ref{tab:empirical_summary} summarizes the correspondence between
theoretical predictions and empirical observations.

\begin{table}[htbp]
\caption{Summary of empirical illustrations.}
\label{tab:empirical_summary}
\centering
\begin{tabular}{lll}
\toprule
Phenomenon & Prediction & Observation \\
\midrule
Match distribution      & Poisson$(\lambda)$       & $d_{\mathrm{TV}} = 0.084$ \\
Disparity amplification & Exponential in $m$       & slope $= 0.78$, $p = 0.0016$ \\
Group dominance & Concentration intensifies & 2\% pop.\ $\to$ 17\% alerts \\
Temporal saturation     & All thresholds fail      & 100\% saturation observed \\
Spatial correlation     & Reduces $k_{\mathrm{eff}}$ & Var/Mean $= 1{,}265$ \\
Temporal correlation    & Reduces $k_{\mathrm{eff}}$ & $\tau \approx 15$ days \\
\bottomrule
\end{tabular}
\end{table}

These illustrations confirm that the theoretical framework describes behaviors
that arise naturally when screening mechanisms are applied to real data. The
phenomena are not artifacts of idealized assumptions.


\end{document}